\documentclass[11pt,a4paper]{article}
\usepackage[qm]{qcircuit}

\usepackage{mathtools}
\usepackage{authblk} 
\usepackage{algpseudocode,algorithmicx,algorithm}

\usepackage{mathrsfs}
\usepackage{latexsym,bm}
\usepackage{amsmath,amsfonts,amssymb,amsthm}
\usepackage{extarrows}

\usepackage{graphicx,subfigure,epstopdf,float}
\usepackage{adjustbox}
\usepackage{enumerate,cases,multirow}
\usepackage{makecell}
\usepackage{caption}

\usepackage{longtable,colortbl,arydshln,threeparttable}
\definecolor{mygray}{gray}{.9}

\usepackage{indentfirst}
\usepackage[top=25mm,bottom=20mm,left=25mm,right=20mm]{geometry}
\usepackage{cite}

\usepackage{listings}

\usepackage{makeidx}        
\usepackage{booktabs}
\usepackage[bookmarks,bookmarksnumbered,colorlinks,citecolor=red,linkcolor=red,hyperindex,linktocpage=true]{hyperref}

\newcommand{\ket}[1]{| #1 \rangle} 
\newcommand{\bra}[1]{\langle #1 |} 

\newcommand{\bb}{\boldsymbol}

\def \d {\mathrm{d}}
\def \e {\mathrm{e}}
\def \i {\mathrm{i}}

\newcounter{parentalgorithm}

\makeatother

\newtheorem{theorem}{Theorem}[section]
\newtheorem{lemma}{Lemma}[section]

\newtheorem{definition}{Definition}[section]

\theoremstyle{remark}
\newtheorem{remark}{\bf Remark}[section]

\numberwithin{equation}{section}

\begin{document}

\title{\bfseries Linear combination of Schr\"odingerization for quantum linear systems with optimal matrix-query complexity \thanks{Submitted to the editors \today. The authors are listed in alphabetical order. All authors contributed equally to this work and share the same co-first authorship.}}

\author[1]{Yin Yang\thanks{yangyinxtu@xtu.edu.cn}}
\author[2]{Yue Yu\thanks{terenceyuyue@xtu.edu.cn}}
\author[3]{Long Zhang\thanks{longzhang@smail.xtu.edu.cn}}

\affil[1]{Hunan Research Center of the Basic Discipline Fundamental Algorithmic Theory and Novel Computational Methods, Xiangtan University, Xiangtan 411105, Hunan, China}
\affil[2]{Key Laboratory of Intelligent Computing and Information Processing of Ministry of Education, Xiangtan University, Xiangtan 411105, Hunan, China}
\affil[3]{Hunan Key Laboratory for Computation and Simulation in Science and Engineering, Hunan International Scientific and Technological Innovation Cooperation Base of Computational Science, Xiangtan University, Xiangtan 411105, Hunan, China}

\maketitle

\begin{abstract}
Quantum linear systems algorithms (QLSAs) aim to solve linear systems $A\bb{x}=\bb{b}$ exponentially faster than classical methods under certain conditions. In this work, we develop quantum algorithms for solving linear algebraic equations from an ODE-based perspective. Inspired by the linear combination of Hamiltonian simulation (LCHS) representation in the Fourier approach \cite{Childs2017QLSA}, we express the solution $\bb{x}$ as a linear combination of solutions to a system of linear convection equations, which become Schr\"odinger-type equations with unitary evolutions in the Fourier domain. We refer to this representation as LC-Schr\"odingerization. Based on this result, we construct an LCHS-based quantum algorithm with two LCHS instances: one for time-marching and one for numerical integration. The key construction uses the derivative of a Gaussian-smoothed hat function and recovers the solution over a fixed auxiliary interval. This permits a truncation time independent of the target accuracy and avoids the loss in success probability from selecting a single grid point. Periodization and explicit Fourier coefficients provide the corresponding projection error bounds. Under the stated oracle assumptions and given a constant-factor estimate of the solution norm, direct simulation of the select operators and block preconditioning achieve the optimal matrix-query complexity $\mathcal{O}(\kappa_A\log\frac1\varepsilon)$ without using variable-time amplitude amplification (VTAA). The same upper bound holds for queries to the right-hand-side preparation oracle.
\end{abstract}



\section{Introduction}

Quantum computing is an emerging computational paradigm that has attracted significant attention, primarily due to the discovery of quantum algorithms capable of offering exponential speedups over the best-known classical methods \cite{Nielsen2010, LR2010QuantumAlgebra, Deutsch1992rapid, Shor1997Prime, HHL2009}. Numerous quantum algorithms for scientific computing have been proposed in recent years. One fundamental task that underlies many areas of science and technology is the development of solvers for a linear system of equations
\begin{equation}\label{Linearsystem}
A\bb{x} = \bb{b},
\end{equation}
where $A$ is an $N\times N$ invertible matrix with $N=2^n$ and $\bb{b}\ne\bb{0}$. We write $\ket{b}=\bb{b}/\|\bb{b}\|$ and $\ket{x}=\bb{x}/\|\bb{x}\|$. Throughout, $\|\cdot\|$ denotes the Euclidean norm for vectors and the spectral norm for matrices, unless otherwise specified. Quantum linear systems algorithms (QLSAs) prepare an approximation to the normalized solution state $\ket{x}$, with the HHL algorithm \cite{HHL2009} being a prominent example.

Classical solvers for linear systems typically take time proportional to the number of unknown variables, making them computationally expensive for large systems. This is especially true for systems arising from the numerical discretization of high-dimensional partial differential equations (PDEs). The HHL algorithm, introduced by Harrow, Hassidim, and Lloyd in 2009 \cite{HHL2009}, is the first quantum algorithm proposed for solving linear systems. Under efficient sparse-matrix access and input-state preparation, its gate complexity is $\widetilde{\mathcal{O}}(\log(N) s^2\kappa^2 /\varepsilon )$, where $ \kappa $ is the condition number of the matrix, $ s $ is the sparsity of $ A $, and $ \varepsilon $ is the precision. In contrast, for positive-definite matrices, the conjugate gradient (CG) method has a time complexity of $\mathcal{O}(N s\sqrt{\kappa} \log\frac1\varepsilon)$, which exhibits polynomial growth with respect to $N$. The quantum advantage in the dimension is therefore possible when the sparsity, condition number, input preparation, and required accuracy are favorable, and the desired output is a quantum state or an efficiently measurable property rather than all entries of $\bb{x}$.
To address the limitations of phase estimation in the HHL algorithm, Ambainis \cite{Ambainis2012VTAA} introduced the variable-time amplitude amplification (VTAA) method, which improves the dependence on the condition number, reducing it from quadratic to nearly linear, up to logarithmic factors.
Building on this, Childs et al. \cite{Childs2017QLSA} further improved the algorithm, achieving a nearly linear dependence on the condition number.
To enhance the accuracy dependence, they replaced phase-estimation-based matrix inversion with Fourier or Chebyshev approximations, resulting in a $\text{polylog}(1/\varepsilon)$ dependence. This is similar to the dependence observed in classical methods and gives an exponential improvement in precision dependence compared to the HHL algorithm.
However, the VTAA procedure involves recursive amplitude amplification, which complicates its implementation \cite{Costa2022QLSA}.
To address this issue, alternative approaches based on adiabatic quantum computing (AQC) and related randomization methods have been developed in recent years \cite{Subasi2019AQC}.
An and Lin \cite{An-Lin-2022} proved that by employing an optimally tuned scheduling function, AQC is capable of efficiently solving a quantum linear system problem (QLSP) with a runtime of $ \mathcal{O}(\kappa  \text{poly}\log\frac\kappa\varepsilon) $. Lin and Tong \cite{Lin-Tong-2020} proposed a quantum eigenstate filtering algorithm and introduced Zeno eigenstate filtering when applying it to the QLSP. One can also refer to \cite{Costa2022QLSA} for a review of the literature along this line. Costa et al. prove a discrete adiabatic theorem and use it to achieve the optimal query complexity $\mathcal{O}(\kappa\log\frac1\varepsilon)$, without VTAA or a truncated Dyson-series subroutine. Earlier discrete-time adiabatic results include \cite{DKS1998Adiabatic}. For a comprehensive overview of the significant algorithmic advancements for QLSAs, we refer the reader to \cite{Lin2026QLSASurvey}.

In this article, we aim at developing quantum algorithms for solving linear algebraic equations from the perspective of ODE-based problems. This approach is straightforward when $ A $ is positive definite. In this case, we can consider the solution $\bb{x}$ as the steady-state solution to the system of linear ordinary differential equations (ODEs) $\bb{u}_t = -A \bb{u} + \bb{b}$, as described in \cite{HJZ2024multiscale}. The resulting linear ODE system is then solved by the Schr\"odingerization method introduced in \cite{JLY22SchrShort, JLY22SchrLong}. Schr\"odingerization transforms linear PDEs and ODEs with non-unitary dynamics into Schr\"odinger-type systems via the so-called warped phase transformation that maps the equation into one higher dimension.
The ODE system can also be solved by using the linear combination of Hamiltonian simulation (LCHS) method in \cite{ACL2023LCH2}, which can be viewed as the continuous implementation of the Fourier transform in the Schr\"odingerization method.

For an indefinite Hermitian matrix, the preceding steady-state construction no longer applies directly. However, the inverse matrix $A^{-1}$ can still be represented in the LCHS form from \cite{ACL2023LCH2}, but with a kernel function based on the Fourier approach in \cite{Childs2017QLSA}. While this LCHS form provides the steady-state solution for a system of linear ODEs with coefficient matrix $-A^2$ and source term $A\bb{b}$, corresponding to the least-squares equations $A^2\bb{x} = A\bb{b}$, applying Schr\"odingerization to this ODE system introduces a squared condition number.
Inspired by the similarities in the LCHS forms and recognizing that the LCHS method for linear ODEs can be recast as the Schr\"odingerization method, we express the solution $\bb{x}$ as a linear combination of the solutions of a system of linear convection equations. These equations are Schr\"odinger-type equations with unitary evolutions in the Fourier domain, and we refer to this solution representation as LC-Schr\"odingerization for short.
Based on this result, we develop an LCHS-based quantum algorithm whose Hamiltonian is given by the Schr\"odinger-type equations. Compared to the Fourier approach in \cite{Childs2017QLSA}, which is based on the direct integral discretization of the LCHS form, our linear combination can be interpreted as requiring only two instances of the LCHS: one for time-marching and the other for numerical integration. This structure separates time marching from numerical integration. The unpreconditioned matrix-query bound still exhibits a quadratic dependence on the condition number. However, by employing the block preconditioning technique proposed in \cite{Low2026quantumlinearsystem} instead of the VTAA procedure, we reduce the condition-number dependence from quadratic to linear, conditional on a constant-factor solution-norm estimate.

The key to the precision dependence is the joint choice of the kernel and the recovery procedure. Within our framework, we first consider a Gaussian kernel that reveals the connection between the abstract representation and the ODE formulation. However, its truncation time grows as the square root of the precision logarithm, and recovering the solution at a single auxiliary grid point incurs an additional loss in success probability. We therefore refine this initial kernel choice and the recovery procedure through the following ingredients.
\begin{itemize}
\item \textbf{Kernel construction.} We smooth the hat function $h(p)=(1-|p|)_+$ by convolution with a normalized Gaussian $g_\sigma$ of standard deviation $\sigma$, and set
\[
G_\sigma(p)=\frac{(h*g_\sigma)(p)}{(h*g_\sigma)(0)},\qquad
\zeta_\sigma(p)=-G_\sigma'(p).
\]
This construction keeps $\|\zeta_\sigma\|_{L^2}$ bounded independently of the smoothing width. By decreasing $\sigma$ with the target error, we can keep the truncation time proportional to an upper bound on $\|A^{-1}\|$, independent of the precision.

\item \textbf{Interval recovery.} The function $G_\sigma$ has a positive constant lower bound on $I=[-1/2,1/2]$. For the convection solution initialized by $\zeta_\sigma(p)\bb{b}$, the identity
\[
\int_0^T\bb{v}(t,p)\d t
=G_\sigma(p)\bb{x}-G_\sigma(pI_N+TA)\bb{x}
\]
shows that the same solution vector can be recovered throughout $I$. Retaining all grid points in this interval produces an approximate product state, and the grid-dependent factors in the initial-state and recovery-vector norms cancel. The auxiliary register can then be discarded.

\item \textbf{Periodization and Fourier projection.} Gaussian smoothing gives explicit Fourier coefficients with rapidly decaying tails. We periodize the kernel and use an $L^2$-orthogonal Fourier projection to control the spatial error, requiring only a logarithmically growing frequency cutoff for the chosen smoothing width.

\item \textbf{Direct simulation of the select operators.} We incorporate the time labels into block-diagonal Hamiltonians and simulate each entire select operator directly, avoiding the repeated precision cost of short-time simulations.
\end{itemize}
Together with block preconditioning, these ingredients achieve the \emph{optimal matrix-query complexity}
\[
\mathcal{O}\Big(\kappa_A\log\frac1\varepsilon\Big),
\]
under the stated oracle assumptions and given a constant-factor estimate of the solution norm. Here $\kappa_A$ is the product of the supplied normalization bounds for $A$ and $A^{-1}$. This result matches the optimal worst-case dependence on the condition-number bound and target precision for quantum linear systems \cite{Costa2022QLSA,Low2026quantumlinearsystem}, and is obtained within our LC-Schr\"odingerization framework without using VTAA. The same upper bound holds for queries to the right-hand-side preparation oracle. The optimality claim concerns matrix queries; the auxiliary-state preparation and norm-estimation assumptions are stated separately, and a complete gate count is beyond the present analysis.

Throughout this paper, $f=\mathcal{O}(g)$ and $f=\Omega(g)$ denote asymptotic upper and lower bounds, respectively, up to positive multiplicative constants independent of the varying parameters; $f=\Theta(g)$ means that both bounds hold. The notation $\widetilde{\mathcal{O}}$ additionally suppresses polylogarithmic factors in the relevant parameters, and $\operatorname{polylog}(x)$ denotes a fixed polynomial in $\log x$.

Table~\ref{tab:query_comparison} compares the matrix and right-hand-side query complexities of representative methods. Our LC-Schr\"odingerization construction attains the optimal matrix-query scaling, while retaining the two-level LCHS formulation described above.

\begin{table}[htbp]
\centering
\caption{Worst-case query complexity of representative quantum linear systems algorithms. For comparison, we use $\kappa\ge2$, $0<\varepsilon\le1/2$, and unit matrix norm with tight normalization bounds, so that $\kappa_A=\Theta(\kappa)$. Sparsity and oracle-conversion costs are held constant.}
\label{tab:query_comparison}
\small
\renewcommand{\arraystretch}{1.3}
\begin{adjustbox}{max width=\textwidth}
\begin{tabular}{@{}lll@{}}
\toprule
\multirow{2}{*}{Methods} & \multicolumn{2}{c}{Query complexity} \\
\cmidrule(lr){2-3}
& $A$ & $\bb{b}$ \\
\midrule
HHL \cite{HHL2009}
& $\widetilde{\mathcal{O}}\Big(\frac{\kappa^2}{\varepsilon}\Big)$
& $\mathcal{O}(\kappa)$ \\
LCU (Fourier, ordinary AA) \cite{Childs2017QLSA}
& $\mathcal{O}\Big(\kappa^2\operatorname{polylog}\frac\kappa\varepsilon\Big)$
& $\mathcal{O}\Big(\kappa\log^{1/2}\frac\kappa\varepsilon\Big)$ \\
LCU+VTAA \cite{Childs2017QLSA}
& $\mathcal{O}\Big(\kappa\operatorname{polylog}\frac\kappa\varepsilon\Big)$
& $\mathcal{O}\Big(\kappa\operatorname{polylog}\frac\kappa\varepsilon\Big)$ \\
QSVT (ordinary AA) \cite{Gilyen2019QSVD}
& $\mathcal{O}\Big(\kappa^2\log\frac\kappa\varepsilon\Big)$
& $\mathcal{O}(\kappa)$ \\
RM \cite{Subasi2019AQC}
& $\widetilde{\mathcal{O}}\Big(\frac\kappa\varepsilon\Big)$
& $\widetilde{\mathcal{O}}\Big(\frac\kappa\varepsilon\Big)$ \\
AQC($p$) \cite{An-Lin-2022}
& $\mathcal{O}\Big(\frac\kappa\varepsilon\log\frac\kappa\varepsilon\Big)$
& $\mathcal{O}\Big(\frac\kappa\varepsilon\log\frac\kappa\varepsilon\Big)$ \\
AQC(exp) \cite{An-Lin-2022}
& $\mathcal{O}\Big(\kappa\operatorname{polylog}\frac\kappa\varepsilon\Big)$
& $\mathcal{O}\Big(\kappa\operatorname{polylog}\frac\kappa\varepsilon\Big)$ \\
dAQC+EF \cite{Costa2022QLSA}
& $\mathcal{O}\Big(\kappa\log\frac1\varepsilon\Big)$
& $\mathcal{O}\Big(\kappa\log\frac1\varepsilon\Big)$ \\
KR \cite{dalzell2026shortcutoptimalquantumlinear}
& $\mathcal{O}\Big(\kappa\log\frac1\varepsilon\Big)$
& $\mathcal{O}\Big(\kappa\log\frac1\varepsilon\Big)$ \\
Tunable VTAA \cite{Low2026quantumlinearsystem}$^{\dagger}$
& $\mathcal{O}\Big(\kappa\log\kappa\log\frac{\log\kappa}{\varepsilon}\Big)$
& $\mathcal{O}(\kappa)$ \\
Block preconditioning \cite{Low2026quantumlinearsystem}$^{\dagger}$
& $\mathcal{O}\Big(\kappa\log\frac1\varepsilon\Big)$
& $\mathcal{O}\Big(\kappa\log\frac1\varepsilon\Big)$ \\
\textbf{This work}$^{\dagger}$
& $\mathcal{O}\Big(\kappa\log\frac1\varepsilon\Big)$
& $\mathcal{O}\Big(\kappa\log\frac1\varepsilon\Big)$ \\
\midrule
Lower bounds \cite{Costa2022QLSA,Low2026quantumlinearsystem}
& $\Omega\Big(\kappa\log\frac1\varepsilon\Big)$
& $\Omega(\kappa)$ \\
\bottomrule
\end{tabular}
\end{adjustbox}

\smallskip
\begin{minipage}{\textwidth}
\footnotesize
\textit{Notes.} Counts include inverse and controlled oracle calls, with constant success probability. The notation $\widetilde{\mathcal{O}}$ suppresses logarithmic factors in $\kappa/\varepsilon$, including simulation overheads. AA denotes amplitude amplification; RM, the randomization method; dAQC, discrete adiabatic quantum computing; EF, eigenstate filtering; and KR, kernel reflection. The Fourier LCU and QSVT rows use ordinary AA; the LCU+VTAA improvement is listed separately.

$^{\dagger}$These entries assume a constant-factor solution-norm estimate and exclude its acquisition cost. The Tunable VTAA entry is the worst-case specialization of \cite[Theorem~2]{Low2026quantumlinearsystem}. The block preconditioning entry combines the preconditioner with QSVT-based matrix inversion, as in \cite[Theorem~4, Eq.~(186)]{Low2026quantumlinearsystem}. This work also assumes the auxiliary-state preparation oracles specified in Section~\ref{subsec:complexity_Sch}; their gate costs are not included. The two lower bounds concern the respective oracles in the worst case. Our result matches the matrix-query lower bound.
\end{minipage}
\end{table}

The paper is organized as follows. In Section \ref{sec:LCHSQLSP}, we introduce an abstract framework that expresses the inverse matrix as a linear combination of Hamiltonian simulations characterized by kernel functions. We then demonstrate how the LCHS formula can be transformed into an LCHS with a Hamiltonian defined by a system of linear convection equations in one higher dimension, which corresponds to Schr\"odinger-type equations in the Fourier domain. In Section \ref{sec:implementationSchrQLSP}, we provide detailed implementation procedures, including interval recovery, along with error and query complexity analyses, where we combine our algorithm with the block preconditioning technique from \cite{Low2026quantumlinearsystem} to achieve optimal matrix-query complexity. Numerical experiments comparing the two kernels are presented in Section~\ref{sec:numerical}. Conclusions are given in the final section.

\section{Linear combination of Schr\"odingerization for quantum linear systems} \label{sec:LCHSQLSP}

This section demonstrates that the solution of the linear system can be expressed as a linear combination of the solutions of Schr\"odinger-type equations with unitary evolutions in the Fourier domain. Consequently, following \cite{JLY22SchrShort,JLY22SchrLong}, we refer to this approach as the Linear Combination of Schr\"odingerization (LC-Schr\"odingerization) for quantum linear systems.


We first introduce an abstract framework for solving the quantum linear system problem (QLSP), which represents $A^{-1}$ as a linear combination of Hamiltonian simulations $\e^{-\i H t_i}$, where $t_i \in \mathbb{R}$ and $H$ is a Hermitian matrix defined by the coefficient matrix of the linear system.

\begin{theorem} \label{thm:abstractLCHSQLSA}
Let $ A $ be an invertible Hermitian matrix, and let $\varphi(k)$ be a function defined on $\mathbb{R}$, with its Fourier transform given by $\hat{\varphi}(s) = \int_{\mathbb{R}}  \varphi(k) \e^{-\i k s} \d k$. Assume that
\[\int_0^\infty \hat{\varphi}(\lambda s) \d s = \frac{1}{\lambda}, \qquad 0\ne \lambda \in [\lambda_{\min}, \lambda_{\max}],\]
where $\lambda_{\min}$ and $\lambda_{\max}$ are the minimum and maximum eigenvalues of $A$, respectively. Then the inverse matrix $A^{-1}$ can be represented as a linear combination of Hamiltonian simulation problems:
\begin{equation}\label{abstractinvA}
A^{-1} =  \int_0^{\infty} \int_{\mathbb{R}}  \varphi(k) \e^{-\i k A s} \d k \d s.
\end{equation}
\end{theorem}
\begin{proof}
The result follows by diagonalizing $A=U\Lambda U^\dag$, where $U$ is unitary and $\Lambda=\operatorname{diag}(\lambda_1,\ldots,\lambda_N)$, and applying the assumed scalar identity to each eigenvalue.
\end{proof}

Inspired by the Fourier approach in \cite{Childs2017QLSA}, we first consider the following kernel within our framework:
\begin{equation}\label{varphik}
\varphi(k) = \frac{\i}{\sqrt{2\pi}} k \e^{-k^2/2}, \qquad \hat{\varphi}(s)  = s\e^{-s^2/2},
\end{equation}
which does not impose any requirements on the signs of the eigenvalues of $A$.

It is clear that $\bb{x}$ can be interpreted as the steady-state solution to the following system of ODEs:
\begin{equation}\label{underlyingODE}
\frac{\d \bb{u}}{\d t} = - A^2\bb{u} + A \bb{b}, \quad \bb{u}(0) = \bb{0}.
\end{equation}
One can verify that
\[\|\bb{x} - \bb{u}(T_u)\| \le \varepsilon \|\bb{x}\| \quad \text{if} \quad T_u \ge \frac{\kappa(A^2)}{\|A^2\|} \log \frac{1}{\varepsilon} = \Big(\frac{\kappa}{\|A\|} \Big)^2 \log \frac{1}{\varepsilon} , \]
where the truncated evolution time depends poorly on the condition number $\kappa$ of $A$. In contrast, as shown below, the truncation time in the LCHS representation \eqref{abstractinvA} depends linearly on the condition number.

For the Gaussian choice in \eqref{varphik}, one can also write the truncated integral as
\begin{equation}\label{xT}
\bb{x}_T=\int_0^T\int_{\mathbb{R}}\varphi(k)\e^{-\i kAt}\d k\d t\,\bb{b}
=(I-\e^{-A^2T^2/2})A^{-1}\bb{b}.
\end{equation}
Thus $\|\bb{x}-\bb{x}_T\|\le\varepsilon\|\bb{x}\|$ if $T\ge\|A^{-1}\|\sqrt{2\log\frac1\varepsilon}$, and $\bb{x}_T=\bb{u}(T^2/2)$ for the solution of \eqref{underlyingODE}. This follows from
\[
\int_0^T\lambda t\e^{-(\lambda t)^2/2}\d t
=\frac{1-\e^{-(\lambda T)^2/2}}{\lambda}.
\]
This choice explains the connection with the ODE-based formulation. However, the factor $\sqrt{\log\frac1\varepsilon}$ in the truncation time also enters the normalization of the time-integral LCU. In the implementation considered below, recovering the solution at a single auxiliary grid point introduces a further loss through the norm of the discretized initial kernel. Even if a fixed recovery interval is used, the Gaussian choice with ordinary amplitude amplification gives a bound of order $\kappa_A\log^{3/2}\frac1\varepsilon$ after block preconditioning, rather than the optimal first power of the precision logarithm. This is a limitation of that implementation and its analysis, not a lower bound for all algorithms using a Gaussian kernel. We therefore retain \eqref{varphik} as motivation and use a different kernel for the subsequent theoretical error and query-complexity analysis.

\subsection{A kernel for interval recovery}

For our problem, we introduce the hat function and a normalized Gaussian,
\[
h(p)=(1-|p|)_+,\qquad
g_\sigma(p)=\frac1{\sqrt{2\pi}\sigma}\e^{-p^2/(2\sigma^2)},
\qquad 0<\sigma\le\frac14.
\]
Let
\begin{equation}\label{smoothed_hat_kernel}
a_\sigma=(h*g_\sigma)(0),\qquad
G_\sigma(p)=\frac{(h*g_\sigma)(p)}{a_\sigma},\qquad
\zeta_\sigma(p)=-G_\sigma'(p).
\end{equation}
The normalization gives $G_\sigma(0)=1$. The smoothing width $\sigma$ will be chosen according to the target accuracy, while the recovery interval
\[
I=[-1/2,1/2]
\]
is fixed. Unlike narrowing a Gaussian alone, smoothing the hat function keeps the $L^2$ norm of its derivative uniformly bounded.

\begin{lemma}\label{lem:kernel_bounds}
For $0<\sigma\le1/4$, the functions in \eqref{smoothed_hat_kernel} satisfy
\begin{equation}\label{kernel_norms}
\frac34\le a_\sigma\le1,\qquad
\frac14\le G_\sigma(p)\le\frac43\quad(p\in I),\qquad
\|\zeta_\sigma\|_{L^2(\mathbb{R})}\le2.
\end{equation}
For $|p|\ge1$,
\begin{equation}\label{kernel_spatial_tail}
|G_\sigma(p)|+|\zeta_\sigma(p)|
\le\frac C\sigma\e^{-(|p|-1)^2/(2\sigma^2)}.
\end{equation}
Moreover, with $\operatorname{sinc}(z)=\sin(z)/z$ and $\operatorname{sinc}(0)=1$, the kernel
\begin{equation}\label{varphi_sigma}
\varphi_\sigma(k)=\frac{\i k}{2\pi a_\sigma}
\operatorname{sinc}^2(k/2)\e^{-\sigma^2k^2/2}
\end{equation}
satisfies $\hat\varphi_\sigma(p)=\zeta_\sigma(p)$ under the Fourier convention of Theorem~\ref{thm:abstractLCHSQLSA}, and
\[
\int_0^\infty\hat\varphi_\sigma(\lambda t)\d t=\frac1\lambda,
\qquad \lambda\ne0.
\]
\end{lemma}
\begin{proof}
Since $h$ is $1$-Lipschitz and $0\le h\le1$, for a standard normal random variable $Z$ we have
\[
|(h*g_\sigma)(p)-h(p)|\le\sigma\mathbb{E}|Z|
=\sigma\sqrt{2/\pi}.
\]
This proves the bounds on $a_\sigma$ and $G_\sigma$ in \eqref{kernel_norms}. The weak derivative of $h$ equals $1$ on $(-1,0)$, $-1$ on $(0,1)$, and zero elsewhere. Young's inequality therefore gives
\[
\|\zeta_\sigma\|_2
\le a_\sigma^{-1}\|h'\|_2\|g_\sigma\|_1
\le\frac{4\sqrt2}{3}<2.
\]
Both $h$ and $h'$ are supported in $[-1,1]$. Bounding the Gaussian in their convolution integrals by its maximum over this support yields \eqref{kernel_spatial_tail}.

The Fourier transform $\int_{\mathbb{R}}h(p)\e^{-\i kp}\d p$ equals $\operatorname{sinc}^2(k/2)$. Convolution with $g_\sigma$ multiplies this transform by $\e^{-\sigma^2k^2/2}$, and differentiation gives \eqref{varphi_sigma}. Finally,
\[
\int_0^T\zeta_\sigma(\lambda t)\d t
=\frac{G_\sigma(0)-G_\sigma(\lambda T)}{\lambda}.
\]
Taking $T\to\infty$ proves the last assertion for both signs of $\lambda$.
\end{proof}

In the following, we show that the LCHS formula in Theorem~\ref{thm:abstractLCHSQLSA} can be converted into the LCHS with the Hamiltonian given by convection equations in one higher dimension. This transformation allows us to apply the discrete Fourier transform, similar to the approach used in the Schr\"odingerization method for ODEs.

\begin{theorem}\label{thm:discreteFourier}
Let $A$ be an invertible Hermitian matrix, and let $\bb{x}=A^{-1}\bb{b}$. Then there holds
\begin{equation}\label{integralt}
\bb{x}=\int_0^\infty\bb{v}(t,0)\d t,
\end{equation}
where $\bb{v}(t,p)$ satisfies the following system of convection equations:
\begin{equation}\label{vtp}
\begin{cases}
\partial_t\bb{v}(t,p)=A\partial_p\bb{v}(t,p),\\
\bb{v}(0,p)=\zeta_\sigma(p)\bb{b}.
\end{cases}
\end{equation}
More generally, for every $p\in\mathbb{R}$ and $T>0$,
\begin{equation}\label{interval_identity}
\int_0^T\bb{v}(t,p)\d t
=G_\sigma(p)\bb{x}-G_\sigma(pI_N+TA)\bb{x}.
\end{equation}
In particular, if $T\ge3\|A^{-1}\|$, then
\begin{equation}\label{interval_time_error}
\sup_{p\in I}\left\|\int_0^T\bb{v}(t,p)\d t-G_\sigma(p)\bb{x}\right\|
\le\frac C\sigma\e^{-9/(8\sigma^2)}\|\bb{x}\|.
\end{equation}
\end{theorem}
\begin{proof}
According to Theorem~\ref{thm:abstractLCHSQLSA} and Lemma~\ref{lem:kernel_bounds}, we have
\[
A^{-1}\bb{b}=\int_0^\infty\int_{\mathbb{R}}
\e^{-\i kAt}(\varphi_\sigma(k)\bb{b})\d k\d t.
\]
Assume that $\bb{v}(t,p)$ is the Fourier transform of
$\check{\bb{v}}(t,k)=\e^{-\i kAt}(\varphi_\sigma(k)\bb{b})$, namely,
\[
\bb{v}(t,p)=\int_{\mathbb{R}}\e^{-\i kp}\check{\bb{v}}(t,k)\d k.
\]
Noting that $\partial_t\check{\bb{v}}=-\i kA\check{\bb{v}}$, we apply the Fourier transform to obtain \eqref{vtp}. It is clear that the solution $\bb{x}$ can be recovered by \eqref{integralt}.

To recover the solution on an interval, write $A=U\operatorname{diag}(\lambda_j)U^\dagger$. The method of characteristics gives
\[
\bb{v}(t,p)=U\operatorname{diag}(\zeta_\sigma(p+\lambda_jt))U^\dagger\bb{b}.
\]
Integrating each component and using $\zeta_\sigma=-G_\sigma'$ proves \eqref{interval_identity}. For $p\in I$, the assumption on $T$ implies $|p+\lambda_jT|\ge5/2$. Applying \eqref{kernel_spatial_tail} gives \eqref{interval_time_error}. This completes the proof.
\end{proof}

\subsection{Periodization and Fourier projection of the new kernel}

The kernel $\zeta_\sigma$ decays rapidly, so an exactly periodic auxiliary problem can be constructed without multiplying it by a cut-off function. For $R\ge2$, set $\Omega_p:=(-R,R)$ and define
\begin{equation}\label{periodic_kernel}
\zeta_R(p)=\sum_{\ell\in\mathbb{Z}}\zeta_\sigma(p+2\ell R).
\end{equation}
Here the dependence of $\zeta_R$ on $\sigma$ is suppressed. The series and all its differentiated series converge uniformly on $\overline{\Omega}_p$, so $\zeta_R$ is smooth and $2R$-periodic.

\begin{lemma}\label{lem:periodization}
Let $A$ be invertible and Hermitian, and let $\bb{v}$ solve \eqref{vtp}. Consider
\begin{equation}\label{perExtension}
\begin{cases}
\partial_t\bb{w}(t,p)=A\partial_p\bb{w}(t,p),\quad p\in\Omega_p,\\
\bb{w}(0,p)=\zeta_R(p)\bb{b},\\
\bb{w}(t,-R)=\bb{w}(t,R).
\end{cases}
\end{equation}
If $0<\sigma\le1/4$ and
\begin{equation}\label{Rnew}
R\ge\|A\|T+2,
\end{equation}
then, with an absolute constant $C$,
\begin{equation}\label{periodization_error}
\sup_{0\le t\le T}\sup_{p\in I}\|\bb{w}(t,p)-\bb{v}(t,p)\|
\le\frac C\sigma\e^{-1/(2\sigma^2)}\|\bb{b}\|.
\end{equation}
Moreover, $\|\zeta_R\|_{L^2(\Omega_p)}\le2$.
\end{lemma}
\begin{proof}
Write $A=U\operatorname{diag}(\lambda_j)U^\dagger$. The method of characteristics gives
\[
\bb{w}(t,p)-\bb{v}(t,p)
=U\operatorname{diag}\left(\sum_{\ell\ne0}\zeta_\sigma(p+\lambda_jt+2\ell R)\right)U^\dagger\bb{b}.
\]
For $p\in I$ and $0\le t\le T$, we have $|p+\lambda_jt|\le R-3/2$. Thus, for $\ell\ne0$,
\[
|p+\lambda_jt+2\ell R|-1\ge(2|\ell|-1)R+1/2.
\]
Using \eqref{kernel_spatial_tail}, the sum over nonzero images is bounded by
\[
\frac C\sigma\sum_{m=1}^\infty
\e^{-((2m-1)R+1/2)^2/(2\sigma^2)}
\le\frac C\sigma\e^{-1/(2\sigma^2)}.
\]
The remaining Gaussian series is uniformly summable for $R\ge2$ and $\sigma\le1/4$. Taking the operator norm proves \eqref{periodization_error}.

For the last assertion, periodize $h'$ and $g_\sigma$ on $\Omega_p$. Their periodic convolution equals $-a_\sigma\zeta_R$. The periodized Gaussian is nonnegative with integral one, while the periodized $h'$ has $L^2$ norm $\sqrt2$. Young's inequality on the periodic domain therefore gives $\|\zeta_R\|_2\le\sqrt2/a_\sigma<2$.
\end{proof}

For the projection, we have the following explicit coefficients and approximation error in the maximum norm.
\begin{lemma}\label{lem:gaussian_fourier}
Let $\mu_k=\pi k/R$ for $k\in\mathbb{Z}$. The Fourier coefficients of $\zeta_R$ in the basis $\e^{\i\mu_kp}$ are
\begin{equation}\label{gaussian_coefficients}
c_k=\frac1{2R}\int_{-R}^R\zeta_R(p)\e^{-\i\mu_kp}\d p
=-\frac{\i\mu_k}{2Ra_\sigma}\operatorname{sinc}^2(\mu_k/2)
\e^{-\sigma^2\mu_k^2/2}.
\end{equation}
Let $N_p=2^{n_p}\ge4$ and $\mathcal{I}=\{-N_p/2,\ldots,N_p/2-1\}$. Denote by $\Pi_p$ the $L^2(\Omega_p)$-orthogonal projection onto $\operatorname{span}\{\e^{\i\mu_kp}:k\in\mathcal{I}\}$. Then
\[
\Pi_p\zeta_R(p)=\sum_{k\in\mathcal{I}}c_k\e^{\i\mu_kp},\qquad
K=\frac\pi R\Big(\frac{N_p}{2}-1\Big).
\]
If $K\ge1$, then
\begin{equation}\label{fourier_tail}
\|\zeta_R-\Pi_p\zeta_R\|_{L^\infty(\Omega_p)}
\le\sum_{k\notin\mathcal{I}}|c_k|
\le C\frac{\e^{-\sigma^2K^2/2}}{\sigma^2K^2}.
\end{equation}
\end{lemma}
\begin{proof}
Unfolding the integral in \eqref{gaussian_coefficients} and applying the Fourier transform of $h*g_\sigma$ gives the stated coefficients. Since
$|u\operatorname{sinc}^2(u/2)|\le4/|u|$, set $d=\pi/R$ and use the monotonicity of $u^{-1}\e^{-\sigma^2u^2/2}$ on $(0,\infty)$ to obtain
\[
\sum_{k\notin\mathcal{I}}|c_k|
\le\frac C R\sum_{k=N_p/2}^\infty\frac{\e^{-\sigma^2(kd)^2/2}}{kd}
\le C\int_K^\infty\frac{\e^{-\sigma^2u^2/2}}u\d u
\le C\frac{\e^{-\sigma^2K^2/2}}{\sigma^2K^2}.
\]
The projection estimate follows by taking absolute values of the omitted Fourier modes.
\end{proof}

\section{Implementation of LC-Schr\"odingerization}\label{sec:implementationSchrQLSP}

We focus solely on Hermitian coefficient matrices. For a non-Hermitian matrix $A$, we introduce the dilation matrix
\[
\widetilde A=\ket{0}\bra{1}\otimes A+\ket{1}\bra{0}\otimes A^\dagger,
\qquad
\widetilde A\begin{bmatrix}\bb{0}\\\bb{x}\end{bmatrix}
=\begin{bmatrix}\bb{b}\\\bb{0}\end{bmatrix}.
\]
The matrix $\widetilde A$ is Hermitian, and its singular values are those of $A$, each repeated twice. Thus the dilation preserves the condition number and the solution norm for the embedded right-hand side.

\subsection{Discretization of the auxiliary variable}

Let $p\in\overline{\Omega}_p=[-R,R]$ and consider the periodic problem \eqref{perExtension}. Then one can apply the Fourier spectral method by discretizing the $p$ domain. Toward this end, we choose a uniform mesh size $\Delta p=2R/N_p$ for the auxiliary variable, with $N_p=2^{n_p}\ge4$. The grid points are denoted by $p_j=-R+j\Delta p$, $j=0,1,\ldots,N_p-1$, with the endpoints identified by periodicity. The grid contains zero at $j=N_p/2$.
Order the modes in $\mathcal{I}$ increasingly and define the unitary Fourier matrix and frequency matrix by
\begin{equation}\label{unitary_fourier}
F_{jk}=\frac1{\sqrt{N_p}}\e^{\i\mu_kp_j},\qquad
D_\mu=\operatorname{diag}(\mu_k)_{k\in\mathcal{I}}.
\end{equation}
For later use, we let
\begin{equation}\label{mu_max_definition}
\mu_{\max}:=\max_{k\in\mathcal{I}}|\mu_k|
=\|D_\mu\|=\frac{\pi N_p}{2R}=\frac{\pi}{\Delta p}.
\end{equation}
The matrix $F$ is implemented by a QFT together with the mode reordering and phase factors associated with the centered grid. All Fourier transformations below use this unitary normalization.

Let $\psi(p)=\zeta_R(p)$ and denote the grid values of its $L^2$-orthogonal Fourier projection by $\bb{\psi}_{\Pi}$. Thus
\begin{equation}\label{projected_initial_state}
\bb{\psi}_{\Pi}
=\bigl((\Pi_p\psi)(p_j)\bigr)_{j=0}^{N_p-1},
\qquad
F^\dagger\bb{\psi}_{\Pi}=(\sqrt{N_p}c_k)_{k\in\mathcal{I}}.
\end{equation}
The operator $\Pi_p$ acts on the periodic function before sampling. Hence $\bb{\psi}_{\Pi}$ contains samples of the projected function, not samples of the unprojected kernel; this distinction avoids an additional aliasing error. By unitarity of $F$ and Parseval's identity,
\begin{equation}\label{norm_psi}
\|\bb{\psi}_{\Pi}\|^2
=\|F^\dagger\bb{\psi}_{\Pi}\|^2
=N_p\sum_{k\in\mathcal{I}}|c_k|^2
\le\frac{N_p}{2R}\|\zeta_R\|_{L^2(\Omega_p)}^2
\le\frac4{\Delta p}.
\end{equation}

Let $\bb{w}(t,p)=[w_1(t,p),\ldots,w_N(t,p)]^\top$ be the solution to \eqref{perExtension}. Its Fourier spectral discretization is given by
\begin{equation}\label{interpcoeff}
\bb{w}_h(t,p)=\sum_{k\in\mathcal{I}}c_k\e^{\i\mu_kp}\e^{\i\mu_kAt}\bb{b}.
\end{equation}
Let the vector $\bb{W}_h$ be the collection of the function $\bb{w}_h$ at the grid points, defined more precisely as
\[
\bb{W}_h(t)=[\bb{w}_h(t,p_0);\ldots;\bb{w}_h(t,p_{N_p-1})],
\]
with ``;'' indicating stacking into a column vector. The equation in \eqref{perExtension} is then transformed into
\begin{equation}\label{wh}
\partial_t\bb{W}_h=\i(FD_\mu F^\dagger\otimes A)\bb{W}_h,
\qquad \bb{W}_h(0)=\bb{\psi}_{\Pi}\otimes\bb{b}.
\end{equation}
In terms of $\widetilde{\bb{W}}_h=(F^\dagger\otimes I)\bb{W}_h$, one gets the following Hamiltonian system:
\begin{equation}\label{discreteLCHSQLSP}
\partial_t\widetilde{\bb{W}}_h=\i(D_\mu\otimes A)\widetilde{\bb{W}}_h,
\qquad \widetilde{\bb{W}}_h(0)=F^\dagger\bb{\psi}_{\Pi}\otimes\bb{b}.
\end{equation}

\subsection{Truncation and projection errors}\label{sec:Truncation of the integral}

The solution to \eqref{discreteLCHSQLSP} can be written as
\[
\widetilde{\bb{W}}_h(t)=\e^{\i(D_\mu\otimes A)t}\widetilde{\bb{W}}_h(0),
\]
which gives
\[
\bb{W}_h(t)=(F\otimes I)\e^{\i(D_\mu\otimes A)t}\widetilde{\bb{W}}_h(0).
\]
According to Theorem~\ref{thm:discreteFourier}, we can recover the approximate solution to the linear systems problem by retaining the grid points in $I=[-1/2,1/2]$ and truncating the integral. Define
\begin{equation}\label{interval_projector}
\Pi_I=\sum_{j:p_j\in I}\ket{j}\bra{j}\otimes I_N,\qquad
\bb{G}_I=(\mathbf{1}_I(p_j)G_\sigma(p_j))_{j=0}^{N_p-1}.
\end{equation}
Then the target vector on these registers is $\bb{G}_I\otimes\bb{x}$, and we approximate it by
\begin{equation}\label{xhT}
\bb{z}_{h,T}=\Pi_I\int_0^T\bb{W}_h(t)\d t
=\Pi_I(F\otimes I)\int_0^T\e^{\i(D_\mu\otimes A)t}\d t\,
(F^\dagger\bb{\psi}_\Pi\otimes\bb{b}).
\end{equation}
The distinction between $\Pi_p$ and $\Pi_I$ is that the former is a Fourier projection on functions, whereas the latter selects a set of auxiliary grid indices.

\begin{theorem}\label{thm:pdiscretization}
Let $0<\varepsilon\le1/2$ and $\xi=\|A^{-1}\ket{b}\|$. Suppose $\alpha_A\ge\|A\|$ and $\alpha_{A^{-1}}\ge\|A^{-1}\|$. Choose
\begin{equation}\label{kernel_parameters}
T=3\alpha_{A^{-1}},\qquad
\sigma=\frac1{2\sqrt{\log\frac{C_0T}{\varepsilon\xi}}},\qquad
R=\alpha_AT+2,
\end{equation}
where $C_0$ is a sufficiently large absolute constant. Choose the smallest power of two $N_p\ge4$ for which
\[
K\ge4\log\frac{C_0T}{\varepsilon\xi}.
\]
Then
\begin{equation}\label{error2}
\|\bb{z}_{h,T}-\bb{G}_I\otimes\bb{x}\|
\le\frac\varepsilon4\|\bb{G}_I\|\|\bb{x}\|.
\end{equation}
These choices satisfy
\begin{equation}\label{projection_parameters}
\begin{aligned}
\mu_{\max}&=\mathcal{O}\Big(\log\frac{T}{\varepsilon\xi}\Big),\qquad
N_p=\mathcal{O}\Big(R\log\frac{T}{\varepsilon\xi}\Big),\\
\|\bb{\psi}_\Pi\|&=\mathcal{O}\Big(\log^{1/2}\frac{T}{\varepsilon\xi}\Big),\qquad
\frac{\|\bb{\psi}_\Pi\|}{\|\bb{G}_I\|}\le12.
\end{aligned}
\end{equation}
\end{theorem}
\begin{proof}
For brevity, write $M=\log\frac{C_0T}{\varepsilon\xi}$. Since $\xi\le\alpha_{A^{-1}}$ and $\varepsilon\le1/2$, choosing $C_0\ge16$ ensures $M\ge4$, $\sigma\le1/4$, and $R\ge5$. Theorem~\ref{thm:discreteFourier} bounds the time truncation error uniformly on $I$ by $C\sqrt M\e^{-9M/2}\|\bb{x}\|$. Lemma~\ref{lem:periodization} bounds the time integral of the periodization error by $CT\sqrt M\e^{-2M}\|\bb{b}\|$. Finally, each $\e^{\i\mu_kAt}$ is unitary, so Lemma~\ref{lem:gaussian_fourier} gives
\begin{equation}\label{errorpk}
\sup_{0\le t\le T}\sup_p\|\bb{w}(t,p)-\bb{w}_h(t,p)\|
\le C\frac{\e^{-\sigma^2K^2/2}}{\sigma^2K^2}\|\bb{b}\|
\le C\frac{\e^{-2M}}M\|\bb{b}\|.
\end{equation}
Consequently,
\[
\sup_{p\in I}\left\|\int_0^T\bb{w}_h(t,p)\d t-G_\sigma(p)\bb{x}\right\|
\le C\sqrt M\e^{-9M/2}\|\bb{x}\|
+CT\sqrt M\e^{-2M}\|\bb{b}\|.
\]
Since $(T/\xi)\e^{-M}=\varepsilon/C_0$ and $\sqrt M\e^{-M}$ is bounded, increasing the absolute constant $C_0$ makes this at most $\varepsilon\|\bb{x}\|/16$. Let $n_I$ denote the number of grid points in $I$. By \eqref{kernel_norms}, $\|\bb{G}_I\|\ge\sqrt{n_I}/4$. Summing the squared pointwise errors over these grid points proves \eqref{error2}.

To make the frequency bound explicit, recall that
\[
K=\frac\pi R\Big(\frac{N_p}{2}-1\Big)=\mu_{\max}-\frac\pi R.
\]
Thus $K\ge4M$ is equivalent to $N_p\ge2+8RM/\pi$. This threshold is greater than four, so the smallest admissible power of two satisfies
\[
N_p=2^{\left\lceil\log_2(2+8RM/\pi)\right\rceil},\qquad
2+\frac{8RM}\pi\le N_p<2\Big(2+\frac{8RM}\pi\Big).
\]
Multiplying by $\pi/(2R)$ gives
\begin{equation}\label{mu_max_scaling}
4M+\frac\pi R\le\mu_{\max}<8M+\frac{2\pi}R=\mathcal{O}(M).
\end{equation}
In particular, $\Delta p=\pi/\mu_{\max}\le\pi/16<1/2$ and $N_p=\mathcal{O}(RM)$. Since the grid contains zero, its number of points in $I$ satisfies
\[
\frac1{2\Delta p}\le n_I\le\frac2{\Delta p}.
\]
Lemma~\ref{lem:periodization}, \eqref{norm_psi}, and \eqref{kernel_norms} now give
\[
\|\bb{\psi}_\Pi\|\le\frac2{\sqrt{\Delta p}},\qquad
\|\bb{G}_I\|\ge\frac1{4\sqrt{2\Delta p}},\qquad
\frac{\|\bb{\psi}_\Pi\|}{\|\bb{G}_I\|}\le8\sqrt2<12.
\]
Since $M=\mathcal{O}(\log\frac{T}{\varepsilon\xi})$, this completes the parameter estimates in \eqref{projection_parameters}.
\end{proof}

\begin{remark}\label{rem:interval_gain}
The projected initial state can be prepared directly in the frequency register, with amplitudes proportional to
\[
-\i\mu_k\operatorname{sinc}^2(\mu_k/2)\e^{-\sigma^2\mu_k^2/2}.
\]
We count this preparation as an auxiliary oracle and do not bound its gate cost here. If raw samples of $\zeta_\sigma$ are used instead, their difference from \eqref{projected_initial_state} must be included in the error budget.

For interval recovery, the grid factors in $\|\bb{\psi}_\Pi\|$ and $\|\bb{G}_I\|$ cancel. This cancellation is lost if only one grid point is retained. For comparison, with the Gaussian derivative kernel associated with our initial choice \eqref{varphik}, one still has this cancellation on a fixed interval, but its truncation time is $\Theta(\|A^{-1}\|\sqrt{\log\frac1\varepsilon})$ and its retained frequency is of order $\sqrt{\log\frac{T}{\varepsilon\xi}}$. Combining the amplitude-amplification and direct-simulation costs as below then gives the $\log^{3/2}\frac1\varepsilon$ bound mentioned after \eqref{xT}. The new kernel also removes the precision dependence from $T$.
\end{remark}

\subsection{Discretization of the truncated integral}\label{sec:Discretization of the truncated integral}

Let $U(A,t)=\e^{\i (D_\mu \otimes A)t}$. We have
\begin{equation}\label{xTtruncate}
\bb{G}_I\otimes\bb{x}\approx\bb{z}_{h,T}=\Pi_I(F\otimes I_N)g(A)\tilde{\bb{W}}_h(0),\qquad g(A)=\int_0^T U(A,t)\d t
\end{equation}
with the parameter choices in Theorem~\ref{thm:pdiscretization}, where $\Pi_I$ is the interval projector in \eqref{interval_projector}.

To present an explicit algorithm, we need to express the truncated integral as a finite sum using numerical integration.
Let $t_m = m \tau$ for $m=0, 1,\cdots,N_t$, where $N_t$ is a power of two and $\tau=T/N_t$ satisfies $\tau\|D_\mu\otimes A\|\le1$. We use composite Gaussian quadrature to discretize the variable $t$ and obtain
\begin{equation}\label{ShrodingerhA}
g(A) = \sum_{m=0}^{N_t-1} \int_{m\tau}^{(m+1)\tau} U(A,t) \d t
\approx \sum_{m=0}^{N_t-1} \sum_{q=0}^{Q-1} w_{q}\tau \e^{\i (D_\mu \otimes A)t_{m,q}} =: p(A),
\end{equation}
where, on each interval $[m \tau, (m + 1)\tau]$, we use Gaussian quadrature with $Q$ nodes. Here $t_{m,q} = t_m + \xi_q \tau$, where the $\xi_q$'s are the Gaussian nodes on $[0,1]$, and the $w_q$'s are the Gaussian weights.

In the implementation, $\e^{\i (D_\mu \otimes A)t_{m,q}}$ will be approximated by some Hamiltonian simulation algorithm as described later. The associated approximation of $p(A)$ will be denoted by $\tilde{p}(A)$ and the approximation for \eqref{xTtruncate} is given by
\begin{equation}\label{xTtruncateGauss}
\bb{z}_{h,T}^{\d} = \Pi_I (F \otimes I_N) \tilde{p}(A) (F^\dagger \otimes I_N) (\bb{\psi}_{\Pi} \otimes \bb{b}).
\end{equation}

To measure the accuracy of the approximation, we first introduce the quadrature error described as follows.

\begin{lemma}\label{lem:ShrodingerGaussErr}
Let $A$ be a Hermitian matrix. For each eigenvalue $\lambda$ of $A$, suppose that $f(\lambda,\cdot)\in C^{2Q}[a,b]$, and define $f(A,x)$ by spectral calculus. Let $w_q$ and $x_q$ be the Gaussian quadrature weights and points on $[a,b]$ for $q=0,1,\cdots,Q-1$. Then there holds
\[\Big\|\int_a^b f(A,x) \d x - \sum_{q=0}^{Q-1} w_q f(A,x_q) \Big\| \le  \frac{(b-a)^{2Q+1}(Q!)^4 }{(2Q+1)[ (2Q)! ]^3} \max_{x\in [a,b]}\|f^{(2Q)} (A,x)\|,\]
where $f^{(2Q)}$ denotes the $2Q$th-order partial derivative with respect to $x$.
\end{lemma}
\begin{proof}
Since $A$ is a Hermitian matrix, there exists a unitary matrix $U$ such that $A=U\Lambda U^\dagger$, where $\Lambda=\operatorname{diag}(\lambda_1,\ldots,\lambda_N)$. Thus the quadrature error has the same norm as the corresponding diagonal error for $f(\Lambda,x)$.
According to the Gaussian quadrature error formula in Chapter 5 of \cite{KMN1989}, its Peano kernel $K_Q$ is nonnegative and satisfies
\[
\int_a^b K_Q(t)\d t=\frac{(b-a)^{2Q+1}(Q!)^4}{(2Q+1)[(2Q)!]^3}.
\]
For each eigenvalue, the integral remainder is
\[
\int_a^b f(\lambda_j,x)\d x-\sum_{q=0}^{Q-1}w_qf(\lambda_j,x_q)
=\int_a^b K_Q(t)f^{(2Q)}(\lambda_j,t)\d t.
\]
This identity applies to complex-valued functions by treating their real and imaginary parts separately. Taking absolute values and then the maximum over $j$ proves the desired estimate.
\end{proof}

\begin{lemma}\label{lem:error1}
Let $p(A)$ be defined in \eqref{ShrodingerhA} and let $0<\varepsilon_1\le T/2$. Under the conditions of Theorem~\ref{thm:pdiscretization}, the quadrature error can be bounded as
\begin{equation}\label{suberror}
\|g(A) -  p(A)\| \le \varepsilon_1
\end{equation}
when we choose
\[\tau\|D_\mu\otimes A\|\le1, \quad Q = \Theta\Big(\log\frac{T}{\varepsilon_1}\Big).\]
\end{lemma}
\begin{proof}
Write the quadrature error as
\[ \|g(A) -  p(A)\| = \left\| g(A) - \sum_{m=0}^{N_t-1} \sum_{q=0}^{Q-1} w_{q}\tau \e^{\i (D_\mu \otimes A)t_{m,q}}\right\| =: I. \]
By Lemma \ref{lem:ShrodingerGaussErr},
\begin{align*}
\left\|  \int_0^{T} U(A,t) \d t - \sum_{m=0}^{N_t-1} \sum_{q=0}^{Q-1}  w_{q}\tau U(A,t_{m,q})\right\|
\le  \frac{N_t \tau^{2Q+1}(Q!)^4 }{(2Q+1)[ (2Q)! ]^3} \max_t \|U^{(2Q)}(A,t) \|.
\end{align*}
Differentiating the evolution operator gives
\[\max_t \|U^{(2Q)}(A,t) \| \le \|D_\mu \otimes A\|^{2Q},\]
which gives
\[I \le T \frac{\tau^{2Q}(Q!)^4 }{(2Q+1)[ (2Q)! ]^3} \|D_\mu \otimes A\|^{2Q}.\]
If we choose $\tau\|D_\mu\otimes A\|\le1$, then
\[I \le \frac{T(Q!)^4 }{(2Q+1)[ (2Q)! ]^3}. \]
Therefore, we obtain $I \le \varepsilon_1$ by choosing $Q$ such that
\[\frac{(Q!)^4 }{(2Q+1)[ (2Q)! ]^3}   \le \frac{\varepsilon_1}{T}.\]
For $Q\ge1$, the inequality $(2Q)!\ge2^Q(Q!)^2$ gives
\[
\frac{(Q!)^4}{(2Q+1)[(2Q)!]^3}
\le\frac{8^{-Q}}{(2Q+1)(Q!)^2}\le8^{-Q}.
\]
Thus any integer $Q\ge\max\{1,\lceil\frac{\log\frac{T}{\varepsilon_1}}{\log 8}\rceil\}$ suffices. Since $T/\varepsilon_1\ge2$, we can choose $Q=\Theta(\log\frac{T}{\varepsilon_1})$. Rounding upwards to a power of two changes this choice by at most a constant factor, which completes the proof.
\end{proof}

\subsection{Linear combination of unitaries}\label{subsec:complexity_Sch}

We retain the interval $\overline{\Omega}_p$ and the unitary Fourier convention in \eqref{unitary_fourier}.
The summation in \eqref{ShrodingerhA} can be simplified as
\begin{align}
p(A) = \sum_{m=0}^{N_t-1} \e^{\i (D_\mu \otimes A) m\tau} \sum_{q=0}^{Q-1} \alpha_q  \e^{\i (D_\mu \otimes A) \xi_q\tau}
=: U_{N_t} U_Q, \label{pALCU}
\end{align}
where $\alpha_q = w_q \tau$.
Our linear combination is interpreted as two instances of the LCHS, corresponding to time marching and numerical integration.

The complexity is determined by the number of queries needed to implement the select oracles associated with the Hamiltonian evolution
\begin{equation}\label{Utau}
U(s) = \e^{\i (D_\mu \otimes A) s }, \qquad 0\le s\le T.
\end{equation}
Since the complexity considered in this paper is the number of queries to the block-encoding oracle of the coefficient matrix $A$, for the convenience of subsequent discussion, we recall the definition of the block-encoding technique with reference to \cite{Gilyen2019QSVD, Chakraborty2019blockEncode, Low2026quantumlinearsystem, An2022blockEncodingODE}.
\begin{definition}\label{def:blockencoding}
Let $A\in\mathbb{C}^{N\times N}$ with $N=2^n$, and set $\Pi=\bra{0^m}\otimes I_N$. For $\alpha>0$ and $\epsilon\ge0$, an $(m+n)$-qubit unitary $U_A$ is an $(\alpha,m,\epsilon)$-block-encoding of $A$ if
\[\|A-\alpha\Pi U_A\Pi^\dag\|\le\epsilon.\]
\end{definition}

We use the following standard Hamiltonian-simulation result, stated for an exact input block-encoding; see Corollary~62 in the full version of \cite{Gilyen2019QSVD}.
\begin{lemma}[Hamiltonian simulation, \cite{Gilyen2019QSVD}]
\label{lem:hamiltonian_simulation}
Let $H$ be Hermitian and let $U_H$ be an $(\alpha_H,a,0)$-block-encoding of $H$. For any $t\in\mathbb{R}$ and $0<\epsilon<1$, one can implement a unitary $V$ that is a $(1,a+2,\epsilon)$-block-encoding of $\e^{\i tH}$, using
\[
\mathcal{O}\!\Big(\alpha_H|t|+\log\frac1\epsilon\Big)
\]
queries to $U_H$ and its inverse, together with $\mathcal{O}(1)$ controlled queries.
\end{lemma}

For simplicity, we suppress the simulation auxiliary registers, which are initialized to zero and left unchanged by the ideal evolution. A block-encoding error $\epsilon$ in the preceding lemma gives an error at most $\sqrt{2\epsilon}$ on such inputs, including amplitude leaving the zero auxiliary state. Thus choosing $\epsilon=\varepsilon_2^2/2$ gives simulation error at most $\varepsilon_2$, with query cost $\mathcal{O}(\alpha_H|t|+\log\frac1{\varepsilon_2})$. All simulation error bounds below refer to inputs with the simulation auxiliary registers initialized to zero.

According to \eqref{pALCU}, the quantum simulation can be interpreted as an LCU at two distinct levels.
In the following we will present the details on how to apply the LCU procedure to effectively prepare the solution state
for $A \bb{x} = \bb{b}$. To this end, we need to pre-construct the following oracles.
\begin{itemize}
\item The state preparation oracles for the coefficients:
\[O_{\text{coef},m} = H^{\otimes n_t} : \ket{0^{n_t}} \to \frac{1}{\sqrt{N_t}} \sum\limits_{m=0}^{N_t-1} \ket{m}, \quad N_t = 2^{n_t},\]
\[O_{\text{coef},q}: \ket{0^{n_Q}} \to \frac{1}{\sqrt{\|\bb{\alpha}\|_1}} \sum\limits_{q=0}^{Q-1}\sqrt{\alpha_q} \ket{q}, \quad Q = 2^{n_Q},\]
with  $\bb{\alpha} = (\alpha_0, \cdots, \alpha_{Q-1})$ and $\|\bb{\alpha}\|_1 = \alpha_0 + \cdots + \alpha_{Q-1}$.
It is obvious that
\[\|\bb{\alpha}\|_1 = \sum_{q=0}^{Q-1} |w_q\tau|
= \tau,\]
where the last equality follows from $\sum_q w_q = 1$.

\item The select oracles
\begin{align}
& \text{SEL}_{A,m} = \sum_{m=0}^{N_t-1} \ket{m}\bra{m} \otimes U(m\tau),  \label{SEL_oracle}\\
& \text{SEL}_{A,q} = \sum_{q=0}^{Q-1} \ket{q}\bra{q} \otimes U(s_q), \quad s_q = \xi_q \tau \le \tau, \label{SEL_oracleQ}
\end{align}
where $U(s)=\e^{\i(D_\mu\otimes A)s}$ is the exact evolution. We implement these entire select operators to error $\varepsilon_2$ and denote their approximations by $\widetilde{\text{SEL}}_{A,m}$ and $\widetilde{\text{SEL}}_{A,q}$. The omitted simulation auxiliary registers are included in the final all-zero success test. The operator $\widetilde p(A)/T$ is the block selected by these zero outcomes and by unpreparing the two coefficient registers. Circuit labels below display the ideal target operators.

\item The state preparation oracles
\begin{align}
O_{\psi}: \ket{0^{n_p}} \to \ket{\psi_{\Pi}},\quad O_b: \ket{0^{n}} \to \ket{b},
\end{align}
where $\ket{\psi_{\Pi}}=\bb{\psi}_{\Pi}/\|\bb{\psi}_{\Pi}\|$ is the projected state in \eqref{projected_initial_state} and $\ket{b}=\bb{b}/\|\bb{b}\|$. One may implement $O_\psi$ by preparing $F^\dagger\ket{\psi_{\Pi}}$ and applying $F$; the following $F^\dagger$ then cancels this transform. Thus the same algorithm can start directly from the explicit Fourier amplitudes $c_k$.
\end{itemize}

\begin{figure}[H]
\centering
\centerline{
\Qcircuit @C=0.45cm @R=0.7cm {
& \lstick{\ket{0^{n_t}}} & \gate{H^{\otimes n_t}} & \qw & \qw & \ctrl{2} & \gate{H^{\otimes n_t}} & \qw & \meter \\
& \lstick{\ket{0^{n_Q}}} & \gate{O_{\mathrm{coef},q}} & \qw & \ctrl{1} & \qw & \gate{O_{\mathrm{coef},q}^\dagger} & \qw & \meter \\
& \lstick{\ket{0^{n_p}}} & \gate{O_\psi} & \gate{F^\dagger} & \multigate{1}{\mathrm{SEL}_{A,q}} & \multigate{1}{\mathrm{SEL}_{A,m}} & \gate{F} & \multigate{2}{O_I} & \qw \\
& \lstick{\ket{0^n}} & \gate{O_b} & \qw & \ghost{\mathrm{SEL}_{A,q}} & \ghost{\mathrm{SEL}_{A,m}} & \qw & \ghost{O_I} & \qw \\
& \lstick{\ket0} & \qw & \qw & \qw & \qw & \qw & \ghost{O_I} & \meter
}}
\caption{Quantum circuit for the overall algorithm (see Eq.~\eqref{xTtruncateGauss}). The interval marker $O_I$ acts on the auxiliary grid and flag registers and as the identity on the system register. Success requires flag outcome one and zero outcomes in the coefficient and simulation registers. The auxiliary grid register is not measured; simulation work registers are omitted.}
\label{fig:U_detailed}
\end{figure}
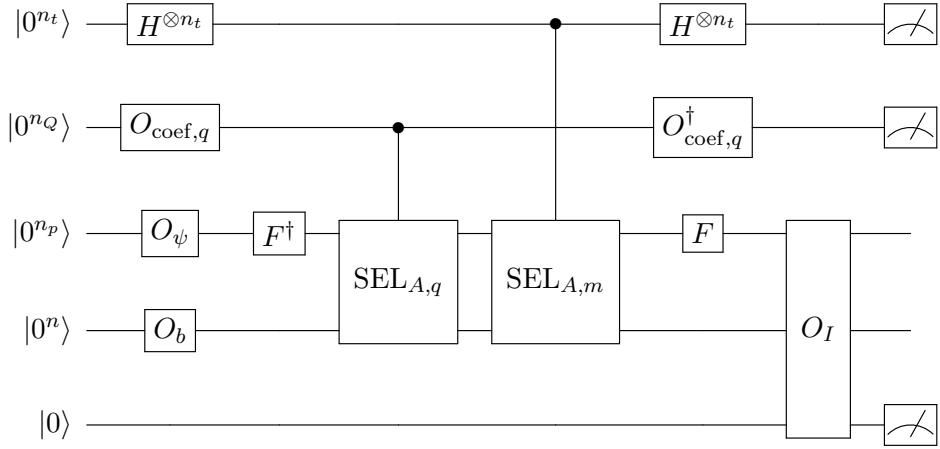

The quantum circuit for the overall algorithm is presented in Fig.~\ref{fig:U_detailed}, with implementation details omitted as the LCU procedures are standard when given the select oracles. Denote by $\mathcal{V}$ the unitary operator up to the final Fourier transform, before interval marking. Let
\[
\bb{W}_{\mathrm{int}}^{\d}
=(F\otimes I)\widetilde p(A)(F^\dagger\bb{\psi}_\Pi\otimes\bb{b}),
\qquad
\ket{W_{\mathrm{int}}^{\d}}=\frac{\bb{W}_{\mathrm{int}}^{\d}}{\|\bb{W}_{\mathrm{int}}^{\d}\|}.
\]
Then,
\begin{equation}\label{W_int}
\mathcal{V}\ket{0^{n_t}}\ket{0^{n_Q}}\ket{0^{n_p}}\ket{0^n}
=\ket{0^{n_t}}\ket{0^{n_Q}}
\frac{\|\bb{W}_{\mathrm{int}}^{\d}\|}{T\|\bb{\psi}_\Pi\|\|\bb{b}\|}
\ket{W_{\mathrm{int}}^{\d}}+\ket\bot.
\end{equation}
Here $\bb{W}_{\mathrm{int}}^{\d}$ approximates $\int_0^T\bb{W}_h(t)\d t$. To obtain the final approximate solution state, we apply the interval marker
\[
O_I:\ket j\ket z\longmapsto\ket j\ket{z\mathbin\oplus\mathbf{1}_I(p_j)}.
\]
The successful component after this operation is
\begin{equation}\label{solution_state}
\ket{0^{n_t}}\ket{0^{n_Q}}
\frac{\|\bb{z}_{h,T}^{\d}\|}{T\|\bb{\psi}_\Pi\|\|\bb{b}\|}
\ket{z_{h,T}^{\d}}\ket1,
\qquad
\ket{z_{h,T}^{\d}}=\frac{\bb{z}_{h,T}^{\d}}{\|\bb{z}_{h,T}^{\d}\|}.
\end{equation}
The remaining component is orthogonal to the success subspace. Theorems~\ref{thm:ShrodingerhAErr} and \ref{thm:time_complexity_Sch} show that $\ket{z_{h,T}^{\d}}$ is close to $\ket{G_I}\ket{x}$, where $\ket{G_I}=\bb{G}_I/\|\bb{G}_I\|$. Thus the auxiliary grid register may be discarded after success. The reversible comparison defining $O_I$ uses only known grid labels and requires no queries to $U_A$ or $O_b$.

The select operators can be implemented directly as evolutions under block-diagonal Hamiltonians. Define the diagonal time-label matrices
\begin{equation}\label{time_label_matrices}
D_m=\sum_{m=0}^{N_t-1}\frac{m}{N_t}\ket{m}\bra{m},
\qquad
D_q=\sum_{q=0}^{Q-1}\xi_q\ket{q}\bra{q}.
\end{equation}
Both have norm at most one. Setting
\[
H_m=D_m\otimes D_\mu\otimes A,\qquad
H_q=D_q\otimes D_\mu\otimes A,
\]
we have the exact identities
\begin{equation}\label{direct_select}
\text{SEL}_{A,m}=\e^{\i T H_m},\qquad
\text{SEL}_{A,q}=\e^{\i\tau H_q}.
\end{equation}
Indeed, the $m$th block of $T H_m$ is $m\tau(D_\mu\otimes A)$, and the $q$th block of $\tau H_q$ is $s_q(D_\mu\otimes A)$. The known diagonal factors $D_m$, $D_q$, and $D_\mu/\mu_{\max}$ admit block-encodings with normalization one, using rotations controlled by their register labels. Tensoring these with $U_A$ gives block-encodings of $H_m$ and $H_q$ with normalization $\alpha_A\mu_{\max}$, each using $\mathcal{O}(1)$ queries to $U_A$. As elsewhere, gates for the known scalar coefficients are not counted as matrix queries; their numerical precision can be chosen within the simulation error budget.

Applying Lemma~\ref{lem:hamiltonian_simulation} to $H_m$ for time $T$ and to $H_q$ for time $\tau$, with the auxiliary-register error conversion described above, gives
\begin{equation}\label{select_error}
\|\widetilde{\text{SEL}}_{A,m}-\text{SEL}_{A,m}\|\le\varepsilon_2,
\qquad
\|\widetilde{\text{SEL}}_{A,q}-\text{SEL}_{A,q}\|\le\varepsilon_2,
\end{equation}
under the same convention of omitted auxiliary registers. The combined matrix-query cost is
\begin{equation}\label{select_cost}
C_{\mathrm{SEL}}
=\mathcal{O}\Big(\alpha_A\mu_{\max}(T+\tau)+\log\frac1{\varepsilon_2}\Big)
=\mathcal{O}\Big(T\alpha_A\mu_{\max}+\log\frac1{\varepsilon_2}\Big).
\end{equation}
Thus $\varepsilon_2$ bounds the error of each entire select operator. There is no accumulation over $N_t$ separately simulated time steps or over $Q$ separately simulated quadrature nodes.

\begin{remark}\label{rem:direct_simulation}
The advantage of direct simulation can be seen by comparing the two implementations at the same target error $\varepsilon_2$. If $U(\tau)$ is approximated with single-step error $\eta$, then repeating this approximation $m$ times gives an error at most $m\eta$. Thus, when constructing $\text{SEL}_{A,m}$ through repeated short-time simulations, choosing $\eta=\varepsilon_2/N_t$ controls the error in all time branches. Applying Lemma~\ref{lem:hamiltonian_simulation} and the auxiliary-register error conversion to each step gives the total matrix-query bound
\[
N_t\,\mathcal{O}\!\Big(\mu_{\max}\alpha_A\tau+\log\frac{N_t}{\varepsilon_2}\Big)
=\mathcal{O}\!\Big(\mu_{\max}\alpha_A T+N_t\log\frac{N_t}{\varepsilon_2}\Big).
\]
In contrast, the identity $\text{SEL}_{A,m}=\e^{\i T H_m}$ allows us to simulate the entire select operator directly with error $\varepsilon_2$, using
\[
\mathcal{O}\!\Big(\mu_{\max}\alpha_A T+\log\frac1{\varepsilon_2}\Big)
\]
matrix queries. The term proportional to the evolution time is the same in both bounds; direct simulation avoids both the stricter single-step tolerance and the repeated precision cost. The partition into $N_t$ intervals is still used for numerical quadrature, but does not require implementing the evolution as $N_t$ approximate short-time steps. Here $\varepsilon_2$ is the select-oracle tolerance; its relation to the final solution error $\varepsilon$ is specified in Theorem~\ref{thm:ShrodingerhAErr}.
\end{remark}

\begin{figure}[htpb]
\centering
\centerline{
\Qcircuit @C=1.0cm @R=1.0cm {
\lstick{\ket{m}} & \multigate{2}{\e^{\i T(D_m\otimes D_\mu\otimes A)}} & \qw \\
\lstick{\ket{k}} & \ghost{\e^{\i T(D_m\otimes D_\mu\otimes A)}} & \qw \\
\lstick{\ket{b}} & \ghost{\e^{\i T(D_m\otimes D_\mu\otimes A)}} & \qw
}}
\caption{Direct Hamiltonian-simulation implementation of $\text{SEL}_{A,m}$. The time label is part of the block-diagonal Hamiltonian. Simulation auxiliary registers are omitted.}
\label{fig:SEL_detailed}
\end{figure}
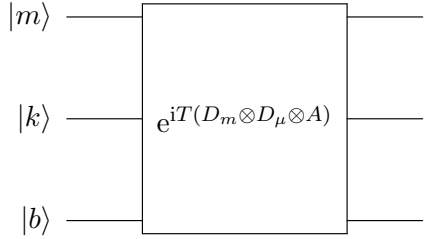

\begin{theorem}\label{thm:ShrodingerhAErr}
Let $A$ be an invertible Hermitian matrix. Suppose that $\bb{x}$ is the exact solution of $A\bb{x}=\bb{b}$ and $\bb{z}_{h,T}^{\d}$ is the numerical vector defined in \eqref{xTtruncateGauss}. Under the conditions of Theorem~\ref{thm:pdiscretization}, choose a power of two $N_t\ge\max\{1,T\mu_{\max}\alpha_A\}$ and set $\tau=T/N_t$. If we choose
\[
\varepsilon_1=\frac{\varepsilon\xi}{128},\qquad
\varepsilon_2=\frac{\varepsilon_1}{4T},
\]
let $Q$ be a power of two satisfying $Q=\Theta(\log\frac{T}{\varepsilon_1})$, with the constant chosen as in Lemma~\ref{lem:error1}, and implement the select operators directly as in \eqref{direct_select} with the accuracy \eqref{select_error}, then there holds
\begin{equation}\label{joint_solution_error}
\|\bb{z}_{h,T}^{\d}-\bb{G}_I\otimes\bb{x}\|
\le\frac\varepsilon2\|\bb{G}_I\|\|\bb{x}\|.
\end{equation}
\end{theorem}
\begin{proof}
First, $\xi\le\alpha_{A^{-1}}$ and $T=3\alpha_{A^{-1}}$ imply
\[
0<\frac{\varepsilon_1}{T}\le\frac{\varepsilon}{384}\le\frac1{768},\qquad
0<\varepsilon_2=\frac{\varepsilon_1}{4T}\le\frac1{3072}.
\]
Thus Lemma~\ref{lem:error1} applies without any additional minimum in the tolerance choices.
The product of the two approximate select operators differs from its ideal counterpart by at most $2\varepsilon_2$. Preparing and unpreparing the coefficient states and projecting onto their zero outcomes do not increase this error. Since the resulting ideal block is $p(A)/T$, we obtain
\begin{equation}\label{hamiltion_error}
\|p(A)-\widetilde p(A)\|\le2T\varepsilon_2.
\end{equation}
The factor $T$ is the total quadrature weight; the direct select implementation introduces no additional factor $N_t$. Because $F$ is unitary and $\|\Pi_I\|=1$, we obtain from Lemma~\ref{lem:error1} that
\begin{equation}\label{errgA}
\|\bb{z}_{h,T}-\bb{z}_{h,T}^{\d}\|
\le(\varepsilon_1+2T\varepsilon_2)\|\bb{\psi}_\Pi\|\|\bb{b}\|.
\end{equation}
By Theorem~\ref{thm:pdiscretization}, $\|\bb{\psi}_\Pi\|/\|\bb{G}_I\|\le12$. Therefore,
\[
\frac{\|\bb{z}_{h,T}-\bb{z}_{h,T}^{\d}\|}{\|\bb{G}_I\|\|\bb{x}\|}
\le\frac{18\varepsilon_1}{\xi}=\frac{18}{128}\varepsilon<\frac\varepsilon4.
\]
The proof is completed by combining \eqref{errgA} and \eqref{error2}.
\end{proof}

\begin{theorem}\label{thm:time_complexity_Sch}
Let $A$ be an invertible Hermitian matrix and assume that $\bb{x}$ is the solution to $A\bb{x}=\bb{b}$. Under the conditions of Theorem~\ref{thm:ShrodingerhAErr}, there exists a quantum algorithm that prepares an $\mathcal{O}(\varepsilon)$-approximation of the state $\ket{x}=\bb{x}/\|\bb{x}\|$, measured in trace distance, with $\Omega(1)$ success probability and a flag indicating success, using
\[
\mathcal{O}\Big(\frac{T}{\xi}\Big)
\]
queries to the coefficient oracles $O_{\mathrm{coef},m}$, $O_{\mathrm{coef},q}$, the select oracles $\mathrm{SEL}_{A,m}$, $\mathrm{SEL}_{A,q}$, and the state preparation oracles $O_b$, $O_\psi$, where $\xi=\|A^{-1}\ket{b}\|$. Interval marking uses no matrix or state-preparation queries.
\end{theorem}
\begin{proof}
Using the inequality $\|\bb{u}/\|\bb{u}\|-\bb{v}/\|\bb{v}\|\|\le2\|\bb{u}-\bb{v}\|/\|\bb{v}\|$ for nonzero vectors and using the estimate in Theorem~\ref{thm:ShrodingerhAErr}, we can bound the error in the joint state after a successful measurement as
\[
\left\|\frac{\bb{z}_{h,T}^{\d}}{\|\bb{z}_{h,T}^{\d}\|}
-\ket{G_I}\ket{x}\right\|
\le\varepsilon,\qquad \ket{G_I}=\frac{\bb{G}_I}{\|\bb{G}_I\|}.
\]
Discarding the auxiliary-variable register therefore gives a state $\rho_x$ with
$\frac12\|\rho_x-\ket{x}\bra{x}\|_1\le\varepsilon$, by contractivity of trace distance under partial trace.

The successful state is obtained by observing all zeros in the coefficient and simulation registers and observing one in the interval flag. The probability of obtaining this approximate state is
\[
\mathrm{P}_r=\left(\frac{\|\bb{z}_{h,T}^{\d}\|}{T\|\bb{\psi}_\Pi\|\|\bb{b}\|}\right)^2.
\]
The success probability can be raised to $\Omega(1)$ by using $\mathcal{O}(g)$ rounds of amplitude amplification, where
\begin{equation}\label{query_times}
g\lesssim\frac{T\|\bb{\psi}_\Pi\|\|\bb{b}\|}{\|\bb{z}_{h,T}^{\d}\|}
\lesssim\frac{T\|\bb{\psi}_\Pi\|}{\xi\|\bb{G}_I\|}
\lesssim\frac T\xi.
\end{equation}
Here \eqref{joint_solution_error} bounds the successful-vector norm from below by $(1-\varepsilon/2)\|\bb{G}_I\|\|\bb{x}\|$, and \eqref{projection_parameters} bounds the ratio of the auxiliary norms. The reflections used in amplitude amplification act on the zero input state and the success flags; they do not require an oracle for $\ket{G_I}$. This completes the proof.
\end{proof}

\subsection{Query complexity of LC-Schr\"odingerization}

\subsubsection{Quadratic dependence on condition number}

Our query complexity result, stated in Theorem~\ref{thm:time_complexity_Sch}, indicates that once the cost of Hamiltonian simulation for the select oracles is taken into account, the overall query complexity exhibits a quadratic dependence on the condition number $\kappa$. In this section, we shall combine the number of queries to the block-encoding $U_A$ of the coefficient matrix $A$ required by the direct implementation of the select oracles, as discussed in \eqref{select_cost}, and then provide the total query complexity needed to recover the original solution.
The following bounds count queries to the matrix block-encoding and right-hand-side preparation oracles. Auxiliary-kernel and quadrature-weight state preparation are treated as the oracles specified above; their gate costs are not included.

\begin{theorem}\label{lem:queries_complexity_Sch}
Let $A$ be an invertible matrix, and consider the linear system $A\bb{x}=\bb{b}$. Suppose that we are given an exact block-encoding of $A$ with normalization factor $\alpha_A\ge\|A\|$ and that an upper bound on its inverse $\alpha_{A^{-1}}\ge\|A^{-1}\|$ is known. Assume access to the auxiliary-state preparation oracles in Section~\ref{subsec:complexity_Sch}. Let
\[
\kappa_A=\alpha_A\alpha_{A^{-1}},\quad
\xi=\|A^{-1}\ket{b}\|.
\]
Here $\kappa_A$ is an upper bound on the condition number $\kappa=\|A\|\|A^{-1}\|$ of $A$. For $0<\varepsilon\le1/2$, there exists a quantum algorithm that prepares an $\mathcal{O}(\varepsilon)$-approximation in trace distance of the state $\ket{x}$ with $\Omega(1)$ success probability and a flag indicating success, using
\begin{equation}\label{costAnew}
\mathcal{O}\left(\frac{\kappa_A^2}{\alpha_A\xi}\log\frac{\kappa_A}{\alpha_A\xi\varepsilon}\right)
\end{equation}
queries to the block-encoding of $A$, and
\begin{equation}\label{costb}
\mathcal{O}\Big(\frac{\kappa_A}{\alpha_A\xi}\Big)
\end{equation}
queries to the preparation oracle for $\bb{b}$. As in the preceding parameter choices, a constant-factor estimate of $\xi$ suffices; its acquisition cost is separate.
\end{theorem}
\begin{proof}
(1) Without loss of generality, we can assume that $A$ is a Hermitian matrix, using the dilation described above when necessary. We first express the costs using $T=3\kappa/\|A\|$ and $R=\|A\|T+2$. For brevity, write
\[
\Lambda_0=\log\frac{\kappa}{\xi\|A\|\varepsilon}.
\]
According to Theorem~\ref{thm:time_complexity_Sch}, we invoke the select oracles and the state preparation oracles $\mathcal{O}(g)$ times, where
\begin{equation}\label{costb_unbounded}
g\lesssim\frac{T\|\bb{\psi}_\Pi\|}{\xi\|\bb{G}_I\|}
\lesssim\frac T\xi
\lesssim\frac{\kappa}{\xi\|A\|}.
\end{equation}
Here we used the interval-recovery estimate in Theorem~\ref{thm:pdiscretization}; the discretization-dependent auxiliary norms cancel.

According to Lemma~\ref{lem:hamiltonian_simulation}, the select oracles $\mathrm{SEL}_{A,m}$ and $\mathrm{SEL}_{A,q}$ defined in \eqref{SEL_oracle} and \eqref{SEL_oracleQ}, implemented directly as in \eqref{direct_select}, have query complexity
\[
C_{\mathrm{SEL}}\lesssim T\mu_{\max}\alpha_A+\log\frac1{\varepsilon_2},
\]
as shown in \eqref{select_cost}. Choose $N_t$ as the smallest power of two at least $T\mu_{\max}\alpha_A$ and set $\tau=T/N_t$. The parameter bounds give
\[
\mu_{\max}=\mathcal{O}(\Lambda_0),\qquad
N_t=\mathcal{O}\Big(\frac{\alpha_A\kappa}{\|A\|}\Lambda_0\Big),\qquad
Q=\Theta\Big(\log\frac{T}{\varepsilon_1}\Big)=\mathcal{O}(\Lambda_0).
\]
The simulation tolerance in Theorem~\ref{thm:ShrodingerhAErr} satisfies
\[
\log\frac1{\varepsilon_2}=\log\frac{512T}{\varepsilon\xi}
=\mathcal{O}(\Lambda_0).
\]
Consequently,
\[
C_{\mathrm{SEL}}\lesssim\frac{\alpha_A\kappa}{\|A\|}\Lambda_0+\Lambda_0
\lesssim\frac{\alpha_A\kappa}{\|A\|}\Lambda_0,
\]
where $\alpha_A/\|A\|\ge1$ and $\kappa\ge1$. Therefore, we use
\begin{equation}\label{costA}
\mathcal{O}(gC_{\mathrm{SEL}})
=\mathcal{O}\left(\frac{\kappa^2\alpha_A}{\xi\|A\|^2}\Lambda_0\right)
\end{equation}
queries to the block-encoding oracle for $A$. The number of queries to the state preparation oracle for $\bb{b}$ is $\mathcal{O}(g)$, as in \eqref{costb_unbounded}.

(2) For brevity, write
\[
\Lambda=\log\frac{\kappa_A}{\alpha_A\xi\varepsilon}.
\]
It is evident that
\[
\kappa=\|A\|\|A^{-1}\|\le\alpha_A\alpha_{A^{-1}}=\kappa_A,\qquad
\frac{\kappa}{\|A\|}=\|A^{-1}\|\le\alpha_{A^{-1}}=\frac{\kappa_A}{\alpha_A}.
\]
Hence, we have
\[
\Lambda_0\le\Lambda,\qquad
\frac{\kappa^2\alpha_A}{\xi\|A\|^2}\le\frac{\kappa_A^2}{\xi\alpha_A},\qquad
\frac{\kappa}{\xi\|A\|}\le\frac{\kappa_A}{\xi\alpha_A}.
\]
These inequalities allow us to replace $\kappa$ and $\kappa/\|A\|$ in \eqref{costA} and \eqref{costb_unbounded} by $\kappa_A$ and $\kappa_A/\alpha_A$, respectively, giving \eqref{costAnew} and \eqref{costb}. The parameter choices can also be made using only these supplied bounds: take $T=3\alpha_{A^{-1}}$, $R=\alpha_AT+2$, and choose $\sigma$, $N_p$, $N_t$, and $Q$ as in Theorems~\ref{thm:pdiscretization} and \ref{thm:ShrodingerhAErr}. Then $\mu_{\max}=\mathcal{O}(\Lambda)$, $g=\mathcal{O}(\kappa_A/(\alpha_A\xi))$, and $C_{\mathrm{SEL}}=\mathcal{O}(\kappa_A\Lambda)$, giving the same bounds without needing the exact norms. A constant-factor estimate of $\xi$ changes only the absolute constants in these choices. The proof is completed.
\end{proof}

\subsubsection{Linear scaling via block preconditioning} \label{sec:preconditioning}

In this section, we demonstrate that the block preconditioning technique introduced in \cite{Low2026quantumlinearsystem} can be applied to achieve linear scaling in the condition number $\kappa_A$.

Let $\Pi_{\text{b}} = \ket{b}\bra{b}$ be the projection operator onto $\ket{b}$. For any $0<s<1$, we define $S = I - (1-s)\Pi_{\text{b}}$, which is referred to as the block preconditioner in \cite{Low2026quantumlinearsystem}, and consider the following preconditioned linear system
\begin{equation}\label{SASb}
SA \bb{x} = S \bb{b}.
\end{equation}
For the quantum computation, we should assume the state preparation oracle
\[O_{Sb}: \ket{0^n} \to  \ket{S b} := \frac{S\bb{b}}{\|S\bb{b}\|}.\]
However, since $\ket{S b} = \ket{b}$, we can take $O_{Sb} = O_b$, where $O_b$ is the state preparation oracle for $\bb{b}$. According to Eq.~(176) in \cite{Low2026quantumlinearsystem}, we can block encode $S$ with normalization factor 1 using two queries to $O_b$. This yields the block-encoding oracle for $SA$ with normalization factor $\alpha_A$.

For the preconditioned system \eqref{SASb}, we summarize its properties below with the details provided in \cite{Low2026quantumlinearsystem}.
\begin{lemma} \label{lem:Sproperty}
Let $A$ be an invertible matrix. Suppose we have a constant multiplicative approximation of the solution norm $\xi = \|A^{-1}\ket{b}\|$, denoted by $\xi_c$, i.e., there exists a constant $c>1$ such that
\[\frac{\xi}{c} < \xi_c < c \xi.\]
Let $\bb{y}=(SA)^{-1}\ket{Sb}$. If we choose $s = \frac{\xi_c}{c\alpha_{A^{-1}}}$, which satisfies
\begin{equation}\label{sChoose}
\frac{\xi}{c^2\alpha_{A^{-1}}} < s < \frac{\xi}{\alpha_{A^{-1}}} \le \frac{\xi}{\|A^{-1}\|} \le 1,
\end{equation}
then we have
\begin{align}
& \xi_{SA}=\|\bb{y}\|=\frac{\xi}{s}, \label{property1}\\
& \|SA\| \le \|A\| \le \alpha_A,\label{property2}\\
&\|(SA)^{-1}\| \le \sqrt{c^4 + 1} \alpha_{A^{-1}}.\label{property3}
\end{align}
\end{lemma}

Combining the above discussion, we obtain the linear dependence.

\begin{theorem}\label{the:finall_op_complexty}
Let $A$ be an invertible matrix and consider the preconditioned linear system \eqref{SASb}, with the exact block-encoding, normalization bounds, and auxiliary-state preparation oracles assumed in Theorem~\ref{lem:queries_complexity_Sch}. For $0<\varepsilon\le1/2$, under the conditions of Lemma~\ref{lem:Sproperty}, including access to the constant-factor estimate $\xi_c$, there exists a quantum algorithm that prepares an $\mathcal{O}(\varepsilon)$-approximation in trace distance of the state $\ket{x}=\ket{y}$, where $\ket{y}=\bb{y}/\|\bb{y}\|$, with $\Omega(1)$ success probability and a flag indicating success, using
\begin{equation}\label{optimal_query_bound}
\mathcal{O}\Big(\kappa_A\log\frac1\varepsilon\Big)
\end{equation}
queries to each of the block-encoding oracle for $A$ and the state-preparation oracle for $\bb{b}$.
\end{theorem}
\begin{proof}
Set $C_c=\sqrt{c^4+1}$. Lemma~\ref{lem:Sproperty} allows the normalization bounds
\[
\alpha_{SA}=\alpha_A,\qquad
\alpha_{(SA)^{-1}}=C_c\alpha_{A^{-1}},\qquad
\kappa_{SA}=C_c\kappa_A.
\]
Since $\xi_{SA}=\xi/s$ and $s=\Theta(\xi/\alpha_{A^{-1}})$,
\[
\frac{\kappa_{SA}}{\alpha_{SA}\xi_{SA}}=\Theta(1),\qquad
\frac{\kappa_{SA}^2}{\alpha_{SA}\xi_{SA}}=\mathcal{O}(\kappa_A).
\]
The logarithmic factor in Theorem~\ref{lem:queries_complexity_Sch}, applied to $SA$, satisfies
\[
\log\frac{\kappa_{SA}}{\alpha_{SA}\xi_{SA}\varepsilon}
=\mathcal{O}\bigl(\log\frac1\varepsilon\bigr),
\]
giving \eqref{optimal_query_bound}. Each query to the block-encoding of $SA$ uses a constant number of queries to $U_A$ and $O_b$, so both query counts obey this bound.
\end{proof}

\begin{remark}
The first power of the precision logarithm combines three features: the new kernel allows $T=3\alpha_{A^{-1}}$; interval recovery keeps $\|\bb{\psi}_\Pi\|/\|\bb{G}_I\|$ bounded; and each entire select oracle is simulated directly as in \eqref{direct_select}. After block preconditioning, $T/\xi_{SA}=\Theta(1)$ with $T=3\alpha_{(SA)^{-1}}$, so $\kappa_A$ also disappears from the precision logarithm. Repeating an approximate short-time simulation $N_t$ times would require a tighter per-step tolerance and would not yield the same estimate by the preceding argument. This distinction concerns matrix queries; the number of register qubits and the gates used to manipulate the time labels can still depend on $\log\kappa_A$.
\end{remark}

\begin{remark}
As shown in \cite{Costa2022QLSA,dalzell2026shortcutoptimalquantumlinear}, the cost of estimating $\xi_c$ in Lemma \ref{lem:Sproperty} can scale linearly with the condition number $\kappa$, under their respective access and success-probability assumptions. The bound above is conditional on this estimate; an end-to-end implementation must add the query cost of the chosen norm-estimation procedure.
\end{remark}

\section{Numerical experiments}\label{sec:numerical}

In this section, we compare the Gaussian kernel in \eqref{varphik} with the Gaussian-smoothed hat kernel in \eqref{smoothed_hat_kernel}. We use UnitaryLab (version 1.1.6) to simulate the two-level LCHS circuit with interval recovery, and use classical spectral calculations to examine grid refinement and time truncation. The select and state preparation oracles are implemented as explicit unitaries for the small matrix below. Thus these are simulations of the algorithm at the oracle level; they do not include QSP approximation, amplitude amplification, or block preconditioning, and do not measure the asymptotic matrix-query complexity.

\subsection{Setup and error measures}

We consider the linear system
\begin{equation}\label{numerical_system}
A=\begin{bmatrix}3/2&-1/2\\-1/2&3/2\end{bmatrix},\qquad
\bb{b}=\begin{bmatrix}1\\0\end{bmatrix},\qquad
\bb{x}=A^{-1}\bb{b}=\begin{bmatrix}3/4\\1/4\end{bmatrix}.
\end{equation}
The eigenvalues of $A$ are $1$ and $2$, so $\kappa=2$. For the Gaussian kernel, we take $G(p)=\e^{-p^2/2}$ and $\zeta(p)=-G'(p)=p\e^{-p^2/2}$. For the smoothed hat kernel, we take $G=G_\sigma$ and $\zeta=\zeta_\sigma$ with $\sigma=1/4$. Both kernels use the same periodic domain, Fourier projection, quadrature, and recovery interval $I=[-1/2,1/2]$. In particular, the Gaussian comparison also uses interval recovery, so that differences between the two kernels are not mixed with differences in the recovery procedure. In addition to \eqref{gaussian_coefficients}, the Gaussian coefficients are
\[
c_k^{\mathrm{G}}=-\frac{\i\mu_k\sqrt{2\pi}}{2R}\e^{-\mu_k^2/2}.
\]
These explicit coefficients are used to prepare the projected initial state, rather than sampling the kernel at grid points.

For the circuit experiments, we set $T=3$, $R=\|A\|T+2=8$, and $Q=4$. For each $N_p$, we choose $N_t$ as the smallest power of two satisfying $N_t\ge T\|A\|\mu_{\max}$. Success requires the two coefficient registers to be zero and the interval flag to be one. The probabilities are obtained directly from the simulated state vector, before amplitude amplification.

Let $\bb{W}_j$ denote the time-integrated grid vector reconstructed from the successful coefficient-register branch using its known LCU and initial-state normalization factors. We report the relative joint-vector error
\begin{equation}\label{numerical_joint_error}
E_I=\frac{\Big(\sum_{p_j\in I}\|\bb{W}_j-G(p_j)\bb{x}\|^2\Big)^{1/2}}
{\|\bb{G}_I\|\|\bb{x}\|},
\end{equation}
the trace distance $D_I=\frac12\|\rho_I-\ket{x}\bra{x}\|_1$, where
\[
\rho_I=\frac{\sum_{p_j\in I}\bb{W}_j\bb{W}_j^\dagger}{\sum_{p_j\in I}\|\bb{W}_j\|^2},
\qquad \ket{x}=\frac{\bb{x}}{\|\bb{x}\|},
\]
and the success probability
\begin{equation}\label{numerical_probability}
P_I=\frac{\sum_{p_j\in I}\|\bb{W}_j\|^2}
{T^2\|\bb{\psi}_\Pi\|^2\|\bb{b}\|^2}.
\end{equation}
The classical solution is used only to assess errors; it is not used to rescale the numerical output. All complex amplitudes are retained, and $\rho_I$ is obtained by tracing out the auxiliary grid register.

\subsection{Kernel comparison and interval recovery}

Table~\ref{tab:kernel_circuits} gives the UnitaryLab results. The total qubit count includes the system, Fourier, time-marching, quadrature, and interval-flag registers. Across the six runs, the full recovered grid vectors agree with an independent classical evaluation of the same spectral quadrature to a relative discrepancy below $4\times10^{-15}$.

\begin{table}[htbp]
\centering
\caption{UnitaryLab state-vector simulations for \eqref{numerical_system}, with $T=3$, $R=8$, $Q=4$, and $\sigma=1/4$ for the smoothed hat kernel. Both kernels use interval recovery. The errors and probabilities are defined in \eqref{numerical_joint_error}--\eqref{numerical_probability}.}
\label{tab:kernel_circuits}
\small
\begin{adjustbox}{max width=\textwidth}
\begin{tabular}{@{}lrrrlll@{}}
\toprule
Kernel & $N_p$ & $N_t$ & Qubits & $E_I$ & $D_I$ & $P_I$ \\
\midrule
Gaussian & 32 & 64 & 15 & $2.54\times10^{-2}$ & $8.10\times10^{-3}$ & $9.70\times10^{-2}$ \\
Smoothed hat & 32 & 64 & 15 & $4.53\times10^{-3}$ & $2.04\times10^{-5}$ & $3.41\times10^{-2}$ \\
Gaussian & 64 & 128 & 17 & $2.17\times10^{-2}$ & $7.12\times10^{-3}$ & $8.45\times10^{-2}$ \\
Smoothed hat & 64 & 128 & 17 & $1.81\times10^{-5}$ & $8.03\times10^{-8}$ & $3.22\times10^{-2}$ \\
Gaussian & 128 & 256 & 19 & $1.97\times10^{-2}$ & $6.66\times10^{-3}$ & $7.77\times10^{-2}$ \\
Smoothed hat & 128 & 256 & 19 & $1.75\times10^{-11}$ & $2.00\times10^{-12}$ & $3.07\times10^{-2}$ \\
\bottomrule
\end{tabular}
\end{adjustbox}
\end{table}

For the smoothed hat kernel, increasing $N_p$ from $32$ to $128$ reduces $E_I$ from $4.53\times10^{-3}$ to $1.75\times10^{-11}$, while $P_I$ remains near $3\times10^{-2}$. The Gaussian error remains near $2\times10^{-2}$ because increasing the Fourier resolution does not remove its time-truncation error at $T=3$. However, the Gaussian kernel has a larger success probability in this example. Thus the benefit of the new kernel is its smaller truncation error at this fixed time, rather than a uniformly larger recovery probability.

Figure~\ref{fig:kernel_numerics} further examines these effects by classical spectral evaluation, with the time integral evaluated analytically to isolate Fourier and truncation errors. In panel (b), the smoothed hat error eventually reaches its finite-time error level for the fixed $\sigma$. In panel (c), replacing interval recovery by the single point $p=0$ makes the success probability decrease approximately as $N_p^{-1}$. For example, for the smoothed hat kernel at $N_p=1024$, interval recovery has probability $2.92\times10^{-2}$, compared with $6.02\times10^{-4}$ for point recovery. In panel (d), $N_p=1024$ and $R=\|A\|T+2$ are used for each time $T$. The smoothed hat kernel reaches the floating-point error level by $T=4$, whereas the Gaussian error is still $1.25\times10^{-5}$ at $T=5$. These finite-dimensional experiments illustrate the roles of the kernel and recovery interval; the optimal query bound follows from the preceding analysis.

\begin{figure}[htbp]
\centering
\subfigure[Kernel functions, $\sigma=1/4$.]{%
\includegraphics[width=0.48\textwidth]{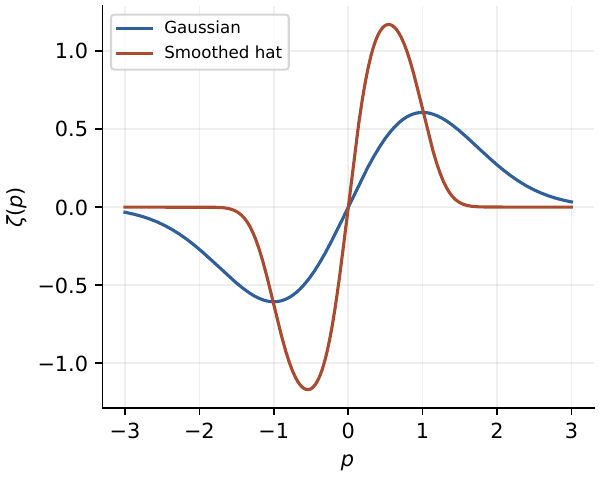}%
\label{fig:kernel_functions}}
\hfill
\subfigure[Fourier resolution, $T=3$.]{%
\includegraphics[width=0.48\textwidth]{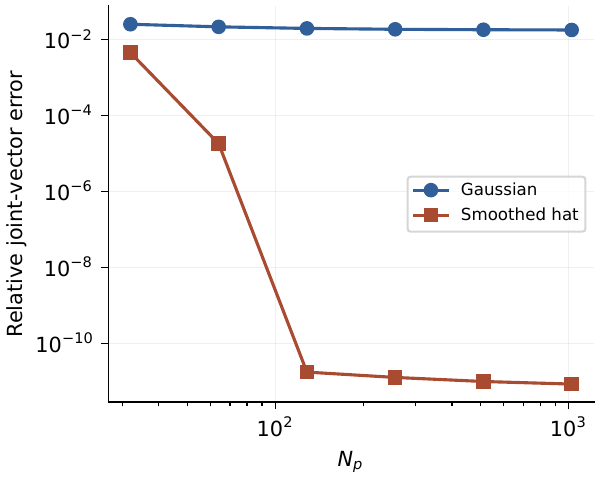}%
\label{fig:fourier_resolution}}

\subfigure[Recovery probability, $T=3$.]{%
\includegraphics[width=0.48\textwidth]{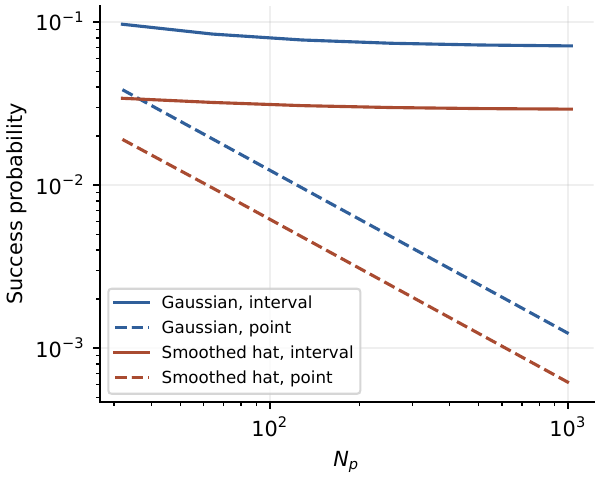}%
\label{fig:recovery_probability}}
\hfill
\subfigure[Time truncation, $N_p=1024$.]{%
\includegraphics[width=0.48\textwidth]{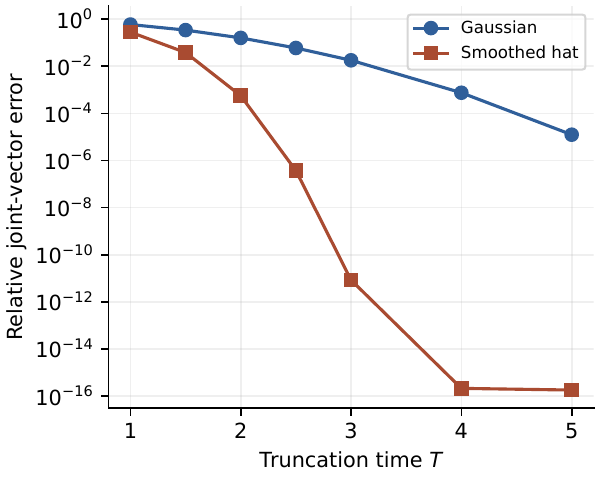}%
\label{fig:time_truncation}}
\caption{Comparison of the Gaussian and smoothed hat kernels for \eqref{numerical_system}. (a) Kernel functions, with $\sigma=1/4$. (b) Relative joint-vector error under Fourier refinement at $T=3$. (c) Success probabilities for interval and point recovery at $T=3$. (d) Dependence on the truncation time with $N_p=1024$. Panels (b)--(d) use classical spectral evaluation with exact time integration; they are separate from the UnitaryLab circuit simulations in Table~\ref{tab:kernel_circuits}. Errors near $10^{-16}$ reflect double-precision roundoff.}
\label{fig:kernel_numerics}
\end{figure}

\clearpage
\section{Conclusions}

In this work, we develop quantum algorithms for solving linear systems of equations from an ODE-based perspective. We introduce an algorithmic framework, which we refer to as LC-Schr\"odingerization. It expresses the solution $\bb{x}=A^{-1}\bb{b}$ as a linear combination of solutions to Schr\"odinger-type equations with unitary evolutions in the Fourier domain. This formulation requires two instances of the linear combination of Hamiltonian simulation (LCHS) method, one for time-marching and one for numerical integration.

The kernel construction and interval recovery are central to the precision dependence. The derivative of a Gaussian-smoothed hat function has a uniformly bounded $L^2$ norm and allows a truncation time independent of the target accuracy. Periodization gives an exactly periodic auxiliary problem without a cut-off function, and the explicit Fourier coefficients control the projection error. Recovering the solution over a fixed interval cancels the grid-dependent normalization factors that reduce the success probability in single-point recovery.

Combining these properties with direct simulation of the select operators and block preconditioning gives the query bound $\mathcal{O}(\kappa_A\log\frac1\varepsilon)$ in Theorem~\ref{the:finall_op_complexty}, without using VTAA. This result is conditional on the stated auxiliary-state oracles and a constant-factor estimate of the solution norm. The UnitaryLab simulations in Section~\ref{sec:numerical} illustrate the accuracy and recovery probability of both kernel choices at the oracle level. A complete gate-level resource analysis, including auxiliary state preparation and the acquisition of the norm estimate, remains a topic for future work.

\section*{Acknowledgments}

The core ideas of this work were proposed by the authors; the underlying LC-Schr\"odingerization framework was presented in their earlier preprint \cite{YangYuZhang2025SchrQLSP}, first posted on arXiv on 19 August 2025. The authors acknowledge the assistance of large language models and AI platforms in selecting the kernel function, working out some details of the mathematical arguments, and polishing the manuscript.

Y. Yang was supported by NSFC grant (No.\ 12571469), Scientific Research Innovation Capability Support Project for Young Faculty of China (No.\ SRICSPYF-BS2025132), the Project of Scientific Research Fund of the Hunan Provincial Science and Technology Department (No.\ 2024JJ1008), the 111 Project (No. D23017), Major Scientific and Technological Innovation Platform Project of Hunan Province (No.\ 2024JC1003), and Program for Science and Technology Innovative Research Team in Higher Educational Institutions of Hunan Province of China.
Y. Yu was supported by NSFC grant (No.\ 12301561), the Key Project of Scientific Research Project of Hunan Provincial Department of Education (No.\ 24A0100), the Science and Technology Innovation Program of Hunan Province (No.\ 2025RC3150) and the general program of Hunan Provincial Natural Science Foundation (No.\ 2026JJ50003), and partially supported by NSFC grant No. 12341104.
L. Zhang was supported by the Hunan Provincial Graduate Student Research and Innovation Project (No.\ CX20250933) and the Xiangtan University Graduate Student Research and Innovation Project (No.\ XDCX2025Y188).

\bibliographystyle{plain} 
\bibliography{Refs}

\end{document}